\documentclass[10pt]{article}
\usepackage[utf8]{inputenc}
\usepackage{amssymb}
\usepackage{amsfonts}
\usepackage{bm}
\usepackage{leftidx}
\usepackage{amsmath}
\usepackage{amsthm}
\usepackage{mathrsfs} 
\usepackage{graphicx}
\usepackage{mathtools}
\usepackage{enumitem}
\usepackage{caption}
\usepackage[nameinlink]{cleveref}
\numberwithin{equation}{section}
\usepackage{xcolor}
\DeclarePairedDelimiter\floor{\lfloor}{\rfloor}
\usepackage{yfonts}
\usepackage{graphicx}
\usepackage{amssymb}
\usepackage{amsmath}
\setlist[itemize]{topsep=0pt,before=\leavevmode\vspace{1.0em}}
\theoremstyle{theorem}
\newtheorem{theorem}{Theorem}[section]
\newtheorem{lemma}[theorem]{Lemma}
\newtheorem{proposition}[theorem]{Proposition}

\newtheorem{definition}[theorem]{Definition}
\newtheorem{remark}[theorem]{Remark}

\newtheorem{assumption}[theorem]{Assumption}
\newtheorem{example}[theorem]{Example}
\newtheorem{corollary}[theorem]{Corollary}

\def\spa{\hskip -3pt}

\def\bearray{\begin{eqnarray}}
\def\earray{\end{eqnarray}}
\def\beq{\begin{equation}}
\def\eeq{\end{equation}}

\def\b0{{\bf 0}}

\newsymbol\bt 1202             

\def\bC{{\mathbb C}}           

\def\bR{{\mathbb R}}

\newsymbol\rest 1316         
\def\gA{{\mathfrak A}}       
\def\gB{{\mathfrak B}}

\begin{document} 

\par
\bigskip
\large
\noindent
{\bf Gibbs states and their classical limit}
\bigskip
\par
\rm
\normalsize

\noindent {\bf Christiaan J.F. van de Ven}\\
\par

\noindent 
Julius-Maximilians-Universit\"{a}t W\"{u}rzburg, Emil-Fischer-Stra\ss e 40, 97074 W\"{u}rzburg, Germany\\
Email: christiaan.vandeven@mathematik.uni-wuerzburg.de\\

\par

\rm\small

\rm\normalsize


\par
\bigskip

\noindent
\small
{\bf Abstract.} 
A continuous bundle of $C^*$-algebras provides a rigorous framework to study the thermodynamic limit of quantum theories. If the bundle admits the additional structure of a strict deformation quantization (in the sense of Rieffel) one is allowed to study the {\em classical limit} of the quantum system, i.e. a mathematical formalism that examines the convergence of algebraic quantum states to probability measures on phase space (typically a Poisson or symplectic manifold). In this manner we first prove the existence of the classical limit of Gibbs states illustrated with a class of Schr\"{o}dinger operators in the regime where Planck's constant $\hbar$ appearing in front of the Laplacian approaches zero. We additionally show that the ensuing limit corresponds to the unique probability measure satisfying the so-called classical or static KMS- condition. Subsequently, we conduct a similar study on the free energy of mean-field quantum spin systems in the regime of large particles, and discuss the existence of the classical limit of the relevant Gibbs states.  Finally, a short section is devoted to single site quantum spin systems in the large spin limit.
\normalsize
\tableofcontents
\newpage

\section{Introduction}\label{Intro}

A modern and rigorous way that establishes a link between classical and quantum theories is based on the theory of {\em quantization},  generally referring to the passage from a classical to a corresponding quantum theory. This goes back to the time when the correct formalism of quantum mechanics was beginning to be discovered. There is in principle no general ``quantization'' recipe working in all cases, and different quantization schemes may lead to inequivalent results with respect to other quantization methods. This is certainly unsatisfactory and depending on the precise purpose, each method has its pros and cons. For example, in geometric quantization (GQ) one aims to obtain a quantum mechanical system given a classical mechanical system whose procedure basically consists of the following three steps: prequantization of the classical system, a polarization method, and finally a metaplectic correction in order to obtain a nonzero quantum Hilbert space \cite{BaWe}. This quantization scheme focuses on the space of states and therefore on the Schr\"{o}dinger picture. A major advantage of GQ is that this technique is very efficient for controlling the physics of the quantum system. Formal deformation quantization (FDQ) instead is based on the construction of the quantum theory via a so-called $\star$ product defined in terms of a formal parameter (typically Planck’s constant $\hbar$). FDQ is useful for e.g. the construction of quantum states in terms of classical ones. The aforementioned quantization procedures are used to obtain quantum mechanics from classical methods. Even though such approaches often give accurate results, they also have their drawbacks: the quantum theory is pre-existing compared with its classical limit and not vice versa. Therefore, one should be able to address the classical limit without the need of imposing a given structure of the quantum model, i.e. that it is obtained as a suitable quantization of a classical one. It is precisely the latter point of view on which this paper is based. Indeed, in this paper quantization is considered as the study of the possible correspondence between a given classical theory, encoded by a Poisson algebra or a Poisson manifold possibly equipped with a (classical) Hamiltonian and flow, and a given quantum theory, mathematically expressed as a certain algebra of observables or a pure state space, possibly equipped with a time evolution and (quantum) Hamiltonian. This does not require at all that quantum theory is formulated in terms of classical structures, and quantization and the classical limit can therefore be seen as equivalent.

Probably the oldest  and best-known example of a pair of a given classical and given quantum theory is classical mechanics of a particle on $\mathbb{R}^n$ with phase space $\bR^{2n}=\{(q,p)\}$ and ensuing $C^*$-algebra of observables given by $\gA_0=C_0(\bR^{2n})$, i.e. the continuous (complex-valued) functions on $\bR^{2n}$ that vanish at infinity, under pointwise operations and supremum norm. Then, the corresponding quantum theory is quantum mechanics with pertinent $C^*$-algebra $\gA_\hbar$ ($\hbar> 0$) taken to be the compact operators $\gB_\infty(L^2(\bR^n))$ on the Hilbert space $L^2(\bR^n)$ for each non-zero $\hbar$. 

Another perhaps less trivial example, originating in the field of mean-field quantum spin systems, is the case for which the classical theory is encoded by the commutative $C^*$-algebra $C(S(M_k(\mathbb{C})))$, i.e. the continuous functions on the algebraic state space of the $(k\times k)$- matrices, containing observables\footnote{These exist under the name of {\em macroscopic} observables.} of an infinite quantum system which describe classical thermodynamics as a limit of mean-field quantum statistical mechanics. The case $k=2$ corresponds to the closed unit $3$-ball $S(M_2(\mathbb{C}))\cong B^3\subset\bR^3$  with  $C(B^3)$ the corresponding $C^*$-algebra. The associated quantum theory is given by the $N$-fold symmetric tensor product of the matrix algebra $M_2(\bC)$ with itself.

The final example we would like mention has been recently discovered \cite{DRVE}. It concerns, as opposed to the previous example, quantum spin systems encoded by local interactions, e.g. the quantum Heisenberg model. In this case, the classical system is defined by the commutative $C^*-$algebra made of equivalence classes generated by averages of local sequences, whilst the corresponding quantum theory of dimension $N$ is the $C^*-$algebra generated by such sequences of length $N$.

A mathematically correct framework that makes the correspondence between both (different) theories precise exists under the name {\em (strict) deformation quantization}, developed in the 1970s (Berezin \cite{Ber} and Bayen et al. \cite{BFFLS}), further elaborated by Rieffel \cite{Rie89,Rie94} and Landsman \cite{Lan98,Lan17}. In this approach the main idea is to ``quantize'' a given
classical (commutative) Poisson algebra into a  given quantum (non-commutative) $C^*$-algebra. In Landsman's approach \cite{Lan17} the starting point of a deformation quantization is often taken to be a {\em continuous bundle of $C^*$-algebras}, which turns out to be highly effective in the study of the classical limit \cite{LMV,MorVen2,Ven2020}.
Probably the most important ingredient in this framework is the notion of the quantization map, whose design can be traced back to Dirac's fundamental ideas on quantum theory and, in view of the previous discussion, it consists of a map $Q_\hbar: \gA_0\supset\tilde{\gA}_0\ni f\mapsto Q_\hbar(f)\in \gA_\hbar$, where $\gA_0 \ (\hbar=0)$ is a commutative $C^*-$algebra containing a dense $*$-Poisson subalgebra $\tilde{\gA}_0$ playing the role of the observable algebra encoding a classical theory, whilst $\gA_\hbar \ (\hbar\neq 0)$ is a non-commutative $C^*$-algebra characteristic for quantum theories.

\subsection{Strict deformation quantization}
To define a strict deformation quantization we take as starting point a continuous-bundle of $C^*-$algebras (see e.g. \cite[Def. IV.1.6.1]{OA2} for the definition). In essence, this is a triple $\mathcal{A}=(I,\gA,\{\pi_\hbar:\gA\to\gA_\hbar\}_{\hbar\in I})$ made by $C^*-$algebras $\gA$, $\gA_\hbar$ where $\hbar$ takes values in a locally compact Hausdorff space $I$, and surjective homomorphisms $\pi_\hbar:\gA\to\gA_\hbar$ satisfying certain continuity conditions. In this way, a continuous section of the bundle is an element $A\in\Pi_{\hbar\in I}\gA_\hbar$ for which  there exists $A'\in\gA$ fulfilling $A_\hbar=\pi_\hbar(A')$ for each $\hbar\in I$.

For the purpose of this paper we quantize a Poisson $*$-algebra $\tilde{\gA}_0$ densely contained in $C_0(X)$ with $X$ a  Poisson manifold. 

\begin{definition}[Def. 7.1 \cite{Lan17}]\label{def:deformationq}
A {\bf deformation quantization} of a Poisson manifold $(X, \{\cdot,\cdot\})$ consists of:
 \begin{itemize}
\item[(1)]  A {continuous  $C^*$-bundle} $\mathcal{A}:=(I, \gA, \{\pi_\hbar:\gA\to\gA_\hbar\}_{\hbar\in I})$,
 where $I$ is a subset of $\mathbb{R}$ containing $0$ as accumulation point and $\gA_0=C_0(X)$ 
equipped with norms $||\cdot||_{\hbar}$;
\item[(2)] a dense $*$-subalgebra $\tilde{\gA}_0$ of  $C_0(X)$ closed under the action of the Poisson brackets
(so that $(\tilde{\gA}_0, \{\cdot, \cdot\})$ is a complex Poisson algebra);
\item[(3)]  a  collection of   {\bf quantization maps} $\{Q_\hbar\}_{\hbar\in I}$, namely linear maps $Q_{\hbar}:\tilde{\gA}_0  \to \gA_{\hbar}$ 
(possibly defined on $\gA_0$ itself and next restricted to $\tilde{\gA}_0$)
such that: 
\begin{enumerate}
\item[(i)] $Q_0$ is the inclusion map $\tilde{\gA}_0 \hookrightarrow \gA_0$ (and $Q_{\hbar}(1\spa1_{\gA_0})=1\spa1_{\gA_{\hbar}}$
if  $\gA_0$ and $\gA_\hbar$ are unital for all $\hbar \in I$);
\item[(ii)] $Q_{\hbar}(\overline{f}) = Q_{\hbar}(f)^*$, where $\overline{f}(x):=\overline{f(x)}$;
\item[(iii)] for each $f\in\tilde{\gA}_0$, the assignments
$
0\mapsto f, \quad 
\hbar\mapsto Q_{\hbar}(f)$ when $\hbar \in I\setminus \{0\},
$
define a continuous section of $\mathcal{A}=(I, \gA, \{\pi_\hbar:\gA\to\gA_\hbar\}_{\hbar\in I})$, meaning that there exists an element $A^f\in\gA$ such that $\pi_\hbar(A^f)=Q_\hbar(f)$ for each $\hbar\in I$.
\item[(iv)]  each pair $f,g\in \tilde{\gA}_0$ satisfies the {\bf Dirac-Groenewold-Rieffel condition}:
\begin{align*}
\lim_{\hbar\to 0}\left|\left|\frac{i}{\hbar}[Q_{\hbar}(f),Q_{\hbar}(g)]-Q_{\hbar}(\{f,g\})\right|\right|_{\hbar}=0.
\end{align*}
\end{enumerate}
\end{itemize}
 If $Q_\hbar(\tilde{\gA}_0)$ is dense in $\gA_\hbar$ for every $\hbar \in I$, then the deformation quantization is called {\bf strict}.\footnote{It immediately follows from the definition of a continuous bundle of $C^*$-algebras that for any $f\in \tilde{\gA}_0$ the  continuity properties, called the {\bf Rieffel condition}, respectively the {\bf von Neumann condition}
\begin{align}\label{Rifc}
\lim_{\hbar\to 0}\|Q_{\hbar}(f)\|_\hbar=\|f\|_{\infty}; \ \ \ \ \lim_{\hbar\to 0}\|Q_{\hbar}(f)Q_{\hbar}(g)-Q_{\hbar}(fg)\|_\hbar=0
\end{align}
automatically hold.}
(If $Q_\hbar$ is defined on the whole $C_0(X)$, all conditions except (iv) are assumed to be valid on $C_0(X)$.)
\hfill $\blacksquare$
\end{definition}
Elements of $I$ are interpreted as possible values of Planck's constant $\hbar$ and $\gA_{\hbar}$ is the quantum algebra of observables of the theory at the given value of $\hbar\neq 0$. 

As mentioned in the first paragraph of the introduction, a standard example of a strict deformation quantization is induced by the sequence $\gA_0=C_0(\bR^{2n}),\ (\hbar=0)$ and $\gA_\hbar=\gB_\infty(\bR^{n}),\ (\hbar>0)$, with  $\gB_\infty(\bR^{n})$ the $C^*-$algebra of compact operators. Note that these algebras are non-unital.

\subsection{Coherent pure state quantization}

In the foregoing we have discussed the concept of quantization theory from the point of view of observables. In that setting we have seen that the quantization map played a crucial role, particularly defining a continuous cross-section of the given $C^*-$bundle. This is not the whole story as quantization maps may be studied in their own right even without the additional structure of a continuous bundle of $C^*-$ algebras.

To do so, we first recall the well-known fact that  the projective Hilbert space $\mathbb{P}\mathcal{H}$ associated to a Hilbert space $\mathcal{H}$ consisting of one-dimensional complex linear subspaces of $\mathcal{H}$ admits the structure of a symplectic manifold with symplectic form denoted by $\Omega_{\mathbb{P}\mathcal{H}}$ \cite[Ch. I]{Lan98}. Moreover, $\mathbb{P}\mathcal{H}$ is equipped with a transition probability 
\begin{align}\label{transitionprob}
    &p:\mathbb{P}\mathcal{H}\times \mathbb{P}\mathcal{H}\to [0,1];&
    &p(\psi,\xi)=|\langle\Psi,\Xi\rangle|^2,
\end{align}
where $\Psi$ and $\Xi$ are arbitrary lifts of $\psi$ and $\xi$ to unit vectors in $\mathcal{H}$. Therefore, the manifold $\mathbb{P}\mathcal{H}$ can be seen as the pure state space of a quantum system.
This can be understood from a algebraic point of view as well. Indeed, if one focuses on $\partial_e S(\gB_\infty(\mathcal{H}))$ i.e. the  extreme boundary of the state space of the algebra of compact operators on a Hilbert space, which in turn can be identified with the convex set of pure states on $\gB_\infty(\mathcal{H})$, it it not difficult to prove that the latter set is isomorphic  to the projective Hilbert space $\mathbb{P}\mathcal{H}$.\footnote{For the algebra of compact operators states $\omega$ bijectively correspond to density matrices $\rho$  via the map $\omega(\cdot)=Tr(\rho \ \cdot)$, i.e.  each state is {\em normal}. In particular, each pure state corresponds to a normal pure state on $\gB(\mathcal{H})$, which in turn identifies a one-dimensional projection  on $\mathcal{H}$.} In view of the previous discussion one may then put $\gA_\hbar=\gB_\infty(\mathcal{H})$.

On the classical side, the pure state space of a classical system is a symplectic manifold $(S,\Omega_S)$, supporting the Liouville measure $\mu_L$ on $S$. Such a classical pure state space may be seen as carrying the ``classical'' transition probability $p_0$, defined by $p_0(\rho,\sigma):=\delta_{\rho,\sigma}$. Again, regarding the previous discussion one may put $\gA_0=C_0(S)$.

This motivates the idea of quantization of a general symplectic manifold $S$ assumed to have finite dimension $2n<\infty$. We successively present the concept of a pure state quantization of a symplectic manifold, the notion the Berezin quantization map associated to a pure state quantization, and the concept of coherent states. 

\begin{definition}[II. Def. 1.3.3 \cite{Lan98}]\label{purestate}
Let $I_0\subset\mathbb{R}$ that has $0\notin I_0$ as an accumulation point and we write $I:=I_0\cup\{0\}$. A {\bf pure state quantization} of a symplectic manifold $(S, \Omega_S)$ consists of a collection of separable Hilbert spaces $\{\mathcal{H}_{\hbar}\}_{\hbar\in I_0}$ and a collection of smooth injections $\{q_\hbar : S \to \mathbb{P}\mathcal{H}_\hbar\}_{\hbar\in I_0}$ for which the following requirements are satisfied. 
 \begin{enumerate}
\item[(1)]  There exists a positive function $c: I_0\to\mathbb{R}\setminus\{0\}$ such that for all $\hbar\in I_0$ and all $\psi\in\mathbb{P}\mathcal{H}_\hbar$ one has
\begin{align}\label{residentity}
c(\hbar)\int_S d\mu_L(\sigma) p(q_\hbar(\sigma), \psi) = 1.
\end{align}
\item[(2)]  For all fixed $f\in C_c(S)$ and $\rho\in S$ the function
\begin{align}
\hbar\mapsto \int_Sd\mu_L(\sigma)p(q_{\hbar}(\rho), q_{\hbar}(\sigma))f(\sigma)
\end{align}
is continuous on $I_0$ and satisfies
\begin{align}\label{resolutionproperty}
\lim_{\hbar\to 0} c(\hbar)\int_Sd\mu_L(\sigma)p(q_{\hbar}(\rho), q_{\hbar}(\sigma))f(\sigma)=f(\rho).
\end{align}
\item[(3)] For each $\hbar\in I_0$ the map $q_{\hbar}$ is an approximate symplectomorphism, in that, (pointwise)
\begin{align}
\lim_{\hbar\to 0}q_{\hbar}^*\Omega_{\mathbb{P}\mathcal{H}_\hbar}=\Omega_S,
\end{align}
where $q_\hbar^*$ is the pull-back.
\end{enumerate}
In the above we have denoted the Liouville measure induced by the symplectic form $\Omega_S$  by $\mu_L$.
\hfill $\blacksquare$
\end{definition}
We would like to point out to the reader that condition (1) in Def. \ref{purestate} defines a so-called {\em resolution of the identity}. Furthermore, it is not difficult to show the satisfying result
\begin{align}
\lim_{\hbar \to 0}p(q_{\hbar}(\rho), q_{\hbar}(\sigma))=p_0(\rho,\sigma)=\delta_{\rho\sigma},
\end{align}
i.e. in quantizing pure states the quantum mechanical transition probabilities converge to the classical ones as $\hbar\to 0$.
In what follows we shall occasionally adopt the short-hand notation $\mu_{\hbar}:=c(\hbar)\mu_L$. 

A pure state quantization naturally leads to the quantization of observables by means of Berezin quantization maps.

\begin{definition}[II. Def. 1.3.4 \cite{Lan98}]\label{Berezinquan}
Let $\{\mathcal{H}_{\hbar}, q_{\hbar}\}_{\hbar\in I_0}$ be a pure state quantization of a symplectic manifold $(S,\Omega_S)$. The {\bf Berezin quantization} of a function $f\in L^{\infty}(S)$ is the family of operators $\{Q^B_{\hbar}
(f)\}_{\hbar\in I_0}$, where $Q^B_{\hbar}(f)\in \gB(\mathcal{H}_\hbar)$ is defined by polarizing
\begin{align}
 \psi(Q^B_{\hbar}(f)):=c(\hbar)\int_S d\mu_L(\sigma)f(\sigma)(q_{\hbar}(\sigma),\psi),
\end{align}
where $\psi\in\mathbb{P}\mathcal{H}_\hbar$.
Convergence of the integral is guaranteed because of \eqref{residentity}. 
\hfill $\blacksquare$
\end{definition}
In case that $f\in L^1(S,d\mu_\hbar)\cap L^{\infty}(S)$, the operator  $Q^B_{\hbar}(f)$ may be written as a Bochner integral
$$Q^B_{\hbar}(f)=c(\hbar)\int_S d\mu_L(\sigma)f(\sigma)[q_{\hbar}(\sigma)],$$
where $[q_{\hbar}(\sigma)]$ is the projection onto the one-dimensional subspace in $\mathcal{H}_\hbar$ whose image in $\mathbb{P}\mathcal{H}_\hbar$ is $q_{\hbar}(\sigma)$. 
 
One can show that the following properties automatically hold. 
\begin{theorem}[II. Thm. 1.3.5 \cite{Lan98}]\label{propertiesberezinmap}
Assume $f\in L^\infty(S,\bR)$.
\begin{itemize}
\item $Q_\hbar^B$ is positive, that is, $f\geq 0$ almost everywhere on $S$ implies $Q_\hbar^B(f)\geq 0$ in $\gB(\mathcal{H}_\hbar)$.
\item $Q_\hbar^B(f)$ is self-adjoint.
\item If $f\in L^1(S,\mu_\hbar)$ then $Q_\hbar^B(f)\in \gB_1(\mathcal{H}_\hbar)$, i.e. $Q_\hbar^B(f)$ is trace-class. Its trace is given by
\begin{align}\label{niceformula}
Tr(Q_\hbar^B(f))=c(\hbar)\int_Sd\mu_L(\sigma)f(\sigma).
\end{align}
\item The operator $Q_\hbar^B(f)$ is bounded by
\begin{align}
||Q_\hbar^B(f)||\leq ||f||_\infty.
\end{align}
\item If $f\in C_0(S)$, then $Q_\hbar^B(f)\in \gB_\infty(\mathcal{H}_\hbar)$,  (i.e $Q_\hbar^B(f)$ is compact) and $Q_\hbar^B:C_0(S)\to \gB_\infty(\mathcal{H}_\hbar)$ is continuous.
\end{itemize}
\end{theorem}

For a given pure state quantization, one may introduce the additional notion of {\em coherent states}. 
\begin{definition}[Def.1.5.1 \cite{Lan98}]\label{coherentsts}
A pure state quantization $\{\mathcal{H}_{\hbar},q_{\hbar}\}_{\hbar\in I_0}$ of $(S,\Omega_S)$ is said to be
{\bf coherent} if each $q_{\hbar}(\sigma)\in \mathbb{P}\mathcal{H}_{\hbar}$ can be lifted to a unit vector $\Psi_{\hbar}^{\sigma}\in\mathcal{H}_{\hbar}$, and the
ensuing map $\sigma\mapsto\Psi_{\hbar}^{\sigma}$ from $S$ to $\mathcal{H}_{\hbar}$ is continuous. The unit vectors  $\Psi_{\hbar}^{\sigma}$ coming
from a coherent pure state quantization are called {\bf coherent states}.
\hfill $\blacksquare$
\end{definition}

Fortunately, in several cases of physical interest the Berezin quantization associated to a coherent pure state quantization of a symplectic manifold also satisfies  the conditions of a strict deformation quantization according to Definition \ref{def:deformationq}.\footnote{For general Poisson manifolds this does not hold since there may not even  exist an additional symplectic structure.} We refer to Appendix \ref{pssm} for two important examples for which this is indeed the case.

\subsection{Classical limit}\label{par:claslim}

A (strict) deformation quantization of a Poisson manifold $X$ naturally leads to the quantization of classical observables. In this fashion, the aforementioned quantization map $Q_\hbar$ which associates (self-adjoint) quantum operators $a\in\gA_\hbar$ to classical observables $f\in\tilde{\gA}_0\subset\gA_0$ plays a crucial role in the study of the  classical limit. In algebraic quantum theory this limit is made rigorous  by means of considering a sequence of {\em algebraic states}\footnote{Positive linear functionals of norm one.} $\omega_\hbar:\gA_\hbar\to \gA_0$, that, depending on the physical situation, may depend on $\hbar$.\footnote{
We stress that $\hbar$ has several interpretations depending on the physical system one considers, e.g. Schr\"{o}dinger operators for which $\hbar$ occurs as Planck's constant ($\S$\ref{SCHrop1}), or quantum spin systems where $\hbar$ plays the role of $1/N$, with $N$ denoting the number of lattice sites (Section \ref{MFQSS}), or the spin quantum number ($\S$\ref{spinnumbernew}), and so on.  Letting $\hbar\to 0$ (provided this limit is taken correctly) should be understood as a way to generate a classical theory, formalized by \eqref{claslim1}.} More precisely, given a sequence of quantization maps $Q_\hbar: \tilde{\gA}_0 \ni f \mapsto Q_\hbar(f) \in \gA_\hbar$, we  say that a sequence of  states $\omega_\hbar:\gA_\hbar\to\mathbb{C}$ is said to be have a {\bf classical limit} if the following limit exists\footnote{We emphasize that this notion of convergence is stronger than the usual approach based on weak-$*$ compactness and the study of converging subsequences. Depending on the situation one can of course weaken the notion of classical limit defined in \eqref{claslim1}.} and defines a state $\omega_0$ on $\tilde{\gA}_0$,
\begin{align}\label{claslim1}
\lim_{\hbar\to 0}\omega_\hbar(Q_\hbar(f))=\omega_0(f), \ (f\in \tilde{\gA}_0),
\end{align}
where $Q_\hbar$ is the quantization map associated with the given (strict  deformation or pure state) quantization. By construction, this approach provides a rigorous meaning of the convergence of algebraic quantum states $\omega_\hbar$  to classical states $\omega_0$ on the commutative algebra  $\tilde{\gA}_0$, when $\hbar \to 0$. The idea behind this formalism is that one is allowed to conveniently exploit the properties of the quantization maps highly suited for studying the semi-classical behavior of various quantum systems \cite{LMV,MorVen2,Ven2020,Ven2021}.  

A special case of interest are the  algebraic quantum (vector) states $\omega_\hbar(\cdot):=\langle \psi_\hbar, (\cdot) \psi_\hbar \rangle$ induced by some normalized unit vector $\psi_\hbar$ forming a sequence in  a Hilbert space on which the observables  $Q_\hbar(f)$ act. The subscript $\hbar$ indicates that the unit vectors might depend on $\hbar$, which is for example the case when $\psi_\hbar$ corresponds to eigenvectors of a $\hbar$-dependent Schr\"{o}dinger operator $H_\hbar$, or in case of spin systems, to eigenvectors $\psi_N$ of a sequence of quantum spin Hamiltonians $H_N$. The main advantage of this $C^*-$algebraic approach is that it typically circumvents convergence problems in Hilbert space. Indeed, the relevant eigenvectors of such operators generally have no limit in the ensuing Hilbert space. These issues have been presented from a technical perspective in \cite{LMV,MorVen2,MV,Ven2020,Ven2021}, where in particular a complete interpretation and rigorous notion of the classical limit of quantum systems have been presented. Furthermore, by taking $\psi_\hbar$ to be a ground state eigenvector of a sequence of quantum Hamiltonians $H_\hbar$, this algebraic approach has shown its efficiency in the study of spontaneous symmetry breaking (SSB) showing up as {\em emergent} phenomenon in the classical limit at zero temperature.

In this paper we extend these ideas to $\beta$-KMS states, initially used to describe quantum states which are in thermal equilibrium at a given (inverse) temperature $\beta$ \cite{Kubo,MarSchw}. Their connection with $C^*$ -dynamical systems has been extensively studied through the years, and $\beta$-KMS states have turned out to be extremely useful in operator algebras as well \cite{BR1,BR2}.  Besides the fact that these states are used in quantum mechanics, also in classical mechanics they have shown major importance. The classical analog of the KMS condition has been introduced in \cite{Katz} and is often studied in the context of infinite classical systems in the continuum. In turns out that the classical KMS condition can be formulated within the context of Poisson and symplectic geometry and naturally leads to the concept of phase transitions \cite{DW}. Given a Poisson manifold $(X,\{\cdot,\cdot\})$ and a vector field $Y$, the set of $\beta$-KMS states $(X,\{\cdot,\cdot\},Y,\beta)$ is a convex set which by construction depends on the choice of $Y$ and $\beta$ \cite{DW}. In this setting, a (classical) phase transition would occur  whenever different choices for $Y,\beta $ produce non-isomorphic convex sets.  In this paper we will not go into these details. We instead first focus on symplectic manifolds of finite dimension and consider a special class of classical $\beta$-KMS functionals, namely Gibbs functionals.\footnote{We stress that a classical $\beta$-KMS functional is not necessary normalizable and therefore does not always define a state.} In some cases these uniquely define the $\beta$-KMS functional for a given dynamics and actually define a state. We examine several cases in quantum as well as classical theories, and finally prove that the classical limit (viz. equation \eqref{claslim1}) for Gibbs states induced by certain quantum Hamiltonians depending on a semi-classical parameter, exists as a  probability measure satisfying the classical (or static) KMS condition. 
Secondly, we extend these ideas to mean-field quantum spin systems in the limit of large particles. We hereto first prove the existence of the limit of the mean-field free energy. This result is not completely ``new'', i.e. similar results are described in \cite{GRV} and the citations herein. However, in such works  the underlying correct algebraic structure of strict deformation quantization establishing the link between classical and quantum theory is not highlighted at all.
Subsequently, we provide a condition which ensures the existence of the classical limit of the Gibbs state. Finally, these ideas are adapted to a study concerning the free energy of a single quantum spin system in the large spin limit.

The paper is structured as follows. In Section \ref{dynamisandKMS} we first introduce the general notion of a $C^*$-dynamical system in quantum and classical theories followed by the definition of Gibbs states in quantum and classical theories (we particularly refer to Prop. \ref{quantum Gibbs}, Def. \ref{def:claskms} and Prop. \ref{classical Gibbs}). Consequently, in Section \ref{claslimsymplecticmanifolds} we use the aforementioned concepts in order to prove a result regarding the convergence of the corresponding quantum Gibbs states to the classical ones (cf. Proposition \ref{generalcase1}). This is illustrated with a class of Schr\"{o}dinger operators (see $\S$\ref{SCHrop1}). In Section \ref{MFQSS} mean-field theories are introduced and emphasized with a well known example.
In Section \ref{The free energy in the regime of large particles} the mean-field free energy associated with mean-field quantum spin systems in the limit of large particles (or lattice sites) is studied. In $\S$\ref{The classical of Gibbs states induced by mean field theories}
the classical limit of the ensuing Gibbs states is discussed, and in $\S$\ref{symmtric case} a special case involving symmetry is addressed. In $\S$\ref{spinnumbernew} a similar study is conducted in the limit of large spin quantum numbers.
Finally, in the appendices useful definitions and technical concepts, particularly adapted to the manifolds $\bR^{2n}$ and $S^2$ necessary for the purpose of this paper are provided.

\section{Dynamics and KMS states in quantum and classical mechanics}\label{dynamisandKMS}
We introduce the general notion of time evolution in algebraic quantum and classical theory followed by some results on quantum and classical KMS states. Let us first recall the definition of a  $C^*$- dynamical system.

\begin{definition}\label{dynamicalsystem}
A $C^*$-{\bf dynamical system} $(\gA, \alpha)$ is a  $C^*$-algebra $\gA$ equipped with a {\bf dynamical evolution}, i.e.,   a one-parameter group of $C^*$-algebra automorphisms $\alpha:= \{\alpha_t\}_{t \in \bR}$ that is {\bf strongly continuous}  on $\gA$:   the map $\bR \ni t \mapsto \alpha_t(A) \in \gA$ is continuous for every $A\in \gA$.
\hfill$\blacksquare$
\end{definition}

In this algebraic context the above definition equally applies to commutative as well as non-commutative $C^*$-algebras.  Let us start on the quantum (non-commutative) side.
\newline
\newline
{\bf Quantum side}\\
Let $t\mapsto U_t$ by a strongly continuous one-parameter group of unitaries acting on a Hilbert space $\mathcal{H}$. As a result of Stone's Theorem, there exists a a self-adjoint operator $H$ densely defined on $\mathcal{H}$ such that $U_t=e^{itH}$, for all $t\in\bR$.  This furthermore induces a one-parameter subgroup of automorphism on $\gB(\mathcal{H})$, namely $\alpha_t(a)=U_t(a)U_t^*$.\footnote{At least on the algebra $\gB(\mathcal{H})$ this construction is invertible, meaning that each strongly continuous one-parameter group of automorphisms $\alpha_t$ on $\gB(\mathcal{H})$ can be obtained in this way.} However, even if the one-parameter group of unitaries on a Hilbert space $\mathcal{H}$ is strongly continuous, it is not necessarily true that the induced one-parameter group of automorphisms $\alpha_t$ on $\gB(\mathcal{H})$ is strongly continuous as well \cite{BR1}, which might happen in the case of unbounded generators. Nonetheless, this condition is satisfied when dealing with the observable algebra of compact operators, even if the self-adjoint generator of $U$ is unbounded \cite[Prop. 6.2]{MorVen2}. 
\\\\
Given a $C^*$-dynamical system $(\gA, \alpha)$ describing a quantum theory, a natural question is to ask how to characterize thermal equilibrium states. The solution relies on a definition introduced by Hugenholtz, Haag and Winnink \cite{HHW} giving a characterization of states satisfying the so-called Kubo Martin Schwinger (KMS) conditions,  firstly studied by Kubo, Martin and Schwinger \cite{Kubo,MarSchw}. In their honour they are therefore  called KMS states.   For sake of completeness let us recall the definition of a KMS state at a given inverse temperature $\beta$, also denoted by KMS$_\beta$-state.

\begin{definition}\label{KMS1}
Consider a $C^*$-dynamical system given by a $C^*$-algebra $\gA$ and a strongly continuous representation $\alpha_t$ of $\bR$ in the automorphism group of $\gA$.
A linear functional $\omega: \gA \to \bC$ is called a {\bf $\boldsymbol{\beta}$-KMS-state} if the following requirements are satisfied:
\begin{itemize}
\item[(1)] $\omega$ is positive, i.e. $\omega(A^*A)\geq 0$ for all $A\in \gA$;
\item[(2)] $\omega$ is normalized, i.e. $\|\omega\|:=\sup \{|\omega(A)| \ : |  \|A\|=1\}=1$;
\item[(3)] $\omega$ satisfies the KMS$_\beta$-condition: 
for all $A,B\in \gA$ there is a holomorphic function $F_{AB}$ on the strip $S_\beta := \bR \times i(0, \beta) \subset \bC$ with a  continuous extension to $\overline{S_\beta}$ such that
$$F_{AB}(t) = \omega(A \alpha_t (B)) \qquad \text{ and } \qquad F_{AB}(t + i\beta) = \omega(\alpha_t (B)A)\,.$$
\end{itemize}
\hfill$\blacksquare$
\end{definition} 
For the purpose of this paper we are interested in an important class of KMS states, namely {\em Gibbs states}.
The reason is that it is not at all clear what generic KMS states look like; this heavily depends on the choice of observable algebra and the given dynamics, which may not assume the standard Heisenberg form. Furthermore, it may happen that the partition function (and hence Gibbs state) is not well defined due to singularities of the pertinent Hamiltonian. The latter is circumvented by imposing additional regularity conditions (cf. Prop. \ref{quantum Gibbs}).

If $\mathcal{H}$ is a Hilbert space and $H$ a self-adjoint linear operator on $\mathcal{H}$, dubbed Hamiltonian, the Gibbs equilibrium state at inverse temperature $\beta$ is defined as a state over $\gB(\mathcal{H})$, i.e. the algebra of bounded operators over $\mathcal{H}$, by \cite{BR2}
\begin{align}\label{quantumgibbsdef}
\omega^\beta(A)=\frac{Tr[e^{-\beta H} A]}{Tr[e^{-\beta H}]},
\end{align}
provided $e^{-\beta H}$ is trace-class. As mentioned above, Gibbs states in general are not the only KMS states at inverse temperature. Nonetheless, the following proposition shows that if the algebra of observables is not too large and the Hamiltonian of physical interest is sufficiently regular, all KMS states  at fixed $\beta$ are uniquely defined in terms of the Gibbs state.

\begin{proposition}[]\label{quantum Gibbs}
For a given Hilbert space $\mathcal{H}$, let $\gA=\gB_\infty(\mathcal{H})$ be the $C^*-$algebra of compact operators over $\mathcal{H}$. Given a self-adjoint Hamiltonian $H$ inducing a (necessarily strongly continuous)  one-parameter group of automorphisms $\alpha^H$ on $\gA$, i.e. $\alpha_t^H(A)=e^{itH}Ae^{-itH}, \ (A\in \gB_\infty(\mathcal{H}))$.  Assume that for any $0<\beta<\infty$ the operator $e^{-\beta H}$ is trace-class.  Then, the following  functional defines the unique $\beta$-KMS state 
on $\gA$ for the  one-parameter group $\alpha^H$,
\begin{align}\label{KMSq}
\omega^{\beta}(\cdot):=Z^{-1}Tr[e^{-\beta H}\cdot],
\end{align}
where $Z$ denotes the partition function associated to $H$. 
\hfill$\blacksquare$
\end{proposition}
\noindent
This result follows from fact that each state on $\gB_\infty(\mathcal{H})$ is normal, i.e. it arises from a density matrix via the trace. See \cite[Ex. 5.3.31]{BR2} for the details.
\\\\
To move on our discussion we now focus on classical theories encoded by commutative $C^*$-algebras of the form $\gA_0:= C_0(X)$.
\newline
\newline
{\bf Classical side}\\
A $C^*$-dynamical system can also be constructed for the commutative $C^*$-algebra $\gA_0:= C_0(X)$ (with $X$ a locally compact Hausdorff space)  equipped with the $C^*$-norm $||\cdot||_\infty$. In particular, if $X$ is a symplectic or a Poisson manifold one can consider the associated Poisson subalgebra $(C^\infty(X), \{\cdot, \cdot\})$ of $\gA$ with Poisson structure denoted by $\{\cdot, \cdot\}$. A $C^*$-dynamical system structure is guaranteed in the case that the dynamical evolution is induced by the pullback action of a complete Hamiltonian flow  $\phi^{h}$, generated by a Hamiltonian function $h\in C^\infty(X)$, i.e.,  $\alpha^{h}_t(f) := f \circ \phi_t^{h}$ for every $f\in C_0(X)$ and $t\in \bR$.   It is easy to show  that $(C_0(X), \alpha^{h})$ is a $C^*$-dynamical system (in particular $\alpha^{h}$ leaves $C_0(X)$ invariant and is strongly continuous). In such cases states $\omega$ correspond to regular Borel probability measures $\mu_\omega$ over $X$, or more generally, to positive measures on $X$. The pure states in turn correspond to  Dirac measures $\delta_\sigma \ (\sigma \in X)$.
\newline
\newline
Analog to the quantum case one may wish to obtain a ``classical'' KMS condition characterizing thermal equilibrium states on a commutative $C^*$-algebra $\gA_0$. To this avail, one typically considers a dense $*$-Poisson subalgebra $\tilde{\gA}_0\subset\gA_0$. In the case that 
$C_c^\infty(X)\subset\gA_0$, the algebra of compactly supported smooth functions on $X$, where $X$ is a Poisson (or symplectic) manifold, such a classical KMS condition can indeed be formalized (see \cite{Katz} and also \cite{AIZ,GALL}). For sake of completeness the definition is given below.

\begin{definition}\label{def:claskms}
 Given a Poisson manifold $X$ together with a vector field $Y\in\Gamma(TS)$, a linear, positive functional $\varphi:C_c^\infty(X)\to \mathbb{C}$ is called a {\bf classical $(\mathbf{Y},\boldsymbol{\beta})$-KMS functional for $\beta>0$} if 
\begin{align}\label{conditionclasKMS}
\varphi(\{f, g\})=\beta\varphi (gY(f)), \ \  \forall f, g\in C_c^\infty(X).
\end{align}
\hfill$\blacksquare$
\end{definition}
\noindent
This condition is also called the static classical KMS condition, and can be extended to a dynamical classical KMS condition in the case where $Y$ has a complete flow \cite{BRW}. For symplectic manifolds it is relatively easy to classify the classical KMS states. Similar to the quantum case (cf. Prop. \ref{quantum Gibbs}), we have the following result.
\begin{proposition}\label{classical Gibbs}
Given a finite-dimensional connected symplectic manifold $(S,\Omega_S)$ together with a Hamiltonian vector field $Y^h\in\Gamma(TS)$, where $h\in C^\infty(S)$ denotes the Hamiltonian function. Assume $e^{-\beta h}\in L^1(S)$. Then, the following positive linear functional is (up to a constant) the unique $(Y^h,\beta)$-KMS functional  for $\beta>0$ on any Poisson $*$-subalgebra $\tilde{\gA}_0$ of $C_0(S)$ containing $C_c^\infty(S)$:
\begin{align}\label{KMSc}
\varphi^\beta(f):=\int_Se^{-\beta h}fd\mu_0, \ (f\in C_0(S));
\end{align}
where $\mu_0$ denotes the Liouville measure on $S$. 
\end{proposition}

\begin{proof}
We first consider the vector space $C_c^\infty(S)$. As a result of \cite[Theorem 4.1]{BRW} or \cite[Remark 4]{DW}, the assignment $C_c^\infty(S)\ni f\mapsto \int_Se^{-\beta h}fd\mu_0$ defines the unique classical $(Y^h,\beta)$-KMS functional  for $\beta>0$, where $Y^h$ denotes the Hamiltonian vector field induced by $h$. Let us denote this functional by $\tilde{\varphi}^\beta$.  Furthermore, as $C^\infty_c(S)\subset C_0(S)$ is dense in the uniform topology, given $f\in C_0(S)$, we can find a sequence $(f_n)_n\subset C^\infty_c(S)$ such that $f_n\to f$ in the uniform norm. Applying the functional $\tilde{\varphi}^\beta$ to this sequence yields
\begin{align}
\tilde{\varphi}^\beta(f_n)=\int e^{-\beta h}f_nd\mu_0.
\end{align}
By the Lebesgue Dominated Convergence Theorem, we conclude
\begin{align}
\varphi^\beta(f):=\lim_{n\to\infty}\tilde{\varphi}^\beta(f_n)=\int e^{-\beta h}fd\mu_0.
\end{align}
Hence the functional $\tilde{\varphi}^\beta$ extends to $C_0(S)$, and this extension is unique by construction. This functional does not satisfy the classical KMS condition \eqref{conditionclasKMS} on all of $C_0(S)$, as $C_0(S)$ has no differentiable structure. Nonetheless, viewing $(S,\Omega_S)$ as a Poisson manifold we can always restrict $\varphi^\beta$ to any Poisson-subalgebra $\tilde{\gA}_0$ of $C_0(S)$ containing $C_c^\infty(S)$. To show that the ensuing functional on $\tilde{\gA}_0$ is the unique $(Y^h,\beta)$-KMS functional  for $\beta>0$, a similar density argument can be applied exploiting the fact that the functional $\tilde{\varphi}^\beta$ on $C^\infty_c(S)$ is the unique $(Y^h,\beta)$-KMS functional  for $\beta>0$.
\end{proof}
In the proof of Proposition \ref{classical Gibbs} we have uniquely extended the positive linear functional $\tilde{\varphi}^\beta$ initially defined on $C_c^\infty(S)$ to all of $C_0(S)$, which we in turn denoted by $\varphi^\beta$. Scaling this functional by $1/c$ where $c:=\int_Se^{-\beta h}d\mu_0$ (this is a finite number due to the fact that $e^{-\beta h}\in L^1(S)$) and using an approximate identity for the $C^*$-algebra $C_0(S)$ one can show that $||\varphi^{\beta\prime}||=1$, where we defined $\varphi^{\beta\prime}:=\varphi^\beta/c$. Since $C_c^\infty(S)\subset\tilde{\gA}_0\subset C_0(S)$ are both dense in $C_0(S)$, the norm $||\varphi^{\beta\prime}||=1$ coincides with its restriction to $C_c^\infty(S)$ and to $\tilde{\gA}_0$. In particular, the restriction $\varphi^{\beta\prime}$ to $\tilde{\gA}_0$ actually defines a state on $\tilde{\gA}_0$. This state is called a {\bf  classical Gibbs state} at inverse temperature $\beta>0$.

\section{The classical limit in the context of symplectic manifolds}\label{claslimsymplecticmanifolds}
In this section we discuss the classical limit of Gibbs states induced by a certain class of  possibly unbounded (self-adjoint) Hamiltonians $H_\hbar$ 
parametrized by a semi-classical parameter $\hbar$.  The Hamiltonians we consider naturally have a classical counterpart  on a symplectic manifold, which again highlights the idea that both quantum theory and classical theory exist in their own right.
More precisely, we assume the following set-up.
\begin{assumption}\label{Assumption}

\begin{itemize}
\item[(i)] Existence of a coherent pure state quantization $\{\mathcal{H}_{\hbar},\Psi_{\hbar}^\sigma\}_{\hbar\in I}$ of a symplectic manifold $(S,\Omega_S)$ with corresponding Hilbert spaces $\{\mathcal{H}_\hbar\}_{\hbar\in I}$ (cf. Def. \ref{purestate}, Def. \ref{coherentsts}).
\item[(ii)]  $H_\hbar$ is bounded below, its domain (on which $H_\hbar$ is essentially self-adjoint) contains the ensuing coherent states $\Psi_\hbar^\sigma$, and $e^{-tH_\hbar}\in\gB_1(\mathcal{H}_\hbar)$ for all $t>0$ and each $\hbar>0$.
\item[(iii)]  Existence of two continuous functions $\check{h}_\hbar$ and $\hat{h}_\hbar$ on $(S,\Omega_S)$ defined by
\begin{align}
\check{h}_\hbar(\sigma)=\langle\Psi_\hbar^\sigma,H_\hbar\Psi_\hbar^\sigma\rangle,
\end{align}
and the (in general non-unique) function $\hat{h}_\hbar$, via the equation
\begin{align}
H_{\hbar}\phi=\int_Sd\mu_\hbar(\sigma)\hat{h}_\hbar(\sigma)\langle\Psi_\hbar^{\sigma},\phi\rangle\Psi_\hbar^{\sigma}, \ \ \phi\in \mathscr{S}(H_\hbar),
\end{align}
such that $\check{h}_\hbar$ and $\hat{h}_\hbar$ both converge pointwise to a continuous function $h_0$ on $(S,\Omega_S)$, and all $\check{h}_\hbar, \hat{h}_\hbar$ and $h_0$ are exponentially integrable, meaning that $f\in C(S)$ satisfies $e^{-tf}\in L_1(S,d\mu_\hbar)\cap  L^\infty(S,d\mu_\hbar)$, for each $t>0$.
\end{itemize}
\end{assumption}
\begin{remark}
{\em 
The function $h_0$ is also called the {\bf principal symbol}, $\hat{h}_\hbar$ the {\bf upper}  or  {\bf contravariant symbol} and $\check{h}_\hbar$ the {\bf lower} or {\bf covariant symbol} associated to the operator $H_\hbar$.
}
\hfill$\blacksquare$
\end{remark}

The above assumption allows us to prove the existence of the classical limit of Gibbs states.

\begin{proposition}\label{generalcase1}
Let $H_\hbar$ be an operator satisfying Assumption \ref{Assumption}. Consider the $\beta$- Gibbs state $\omega_{\hbar}^{\beta}$ ($\beta<\infty$) given by
\begin{align}
\omega_{\hbar}^{\beta}(\cdot)=\frac{Tr[\ \cdot \ e^{-\beta H_\hbar}]}{Tr[e^{-\beta H_\hbar}]}.
\end{align}
Then the following limit exists for any real-valued $f\in C_0(S)$
\begin{align}\label{claslimKMS}
\lim_{\hbar\to 0}\bigg{|}\omega_{\hbar}^{\beta}(Q_{\hbar}^B(f)) - \frac{\int_{S}d\sigma f(\sigma)e^{-\beta h_0(\sigma)}}{\int_{S}d\sigma e^{-\beta h_0(\sigma)}}\bigg{|}=0,
\end{align}
where $d\sigma$ denotes  the Liouville measure on $S$  and $Q_\hbar^B$ is the Berezin quantization map associated with the given coherent pure state quantization.
\end{proposition}

In order to prove the proposition we start with a result relating the classical and quantum partition functions. This result can be seen as a corollary of the so-called Berezin-Lieb inequality (we refer to the books \cite{Gaz} and \cite{Com} for a detailed discussion on this topic). For sake of completeness we state the result by means of the following lemma.
\begin{lemma}\label{berezinlieblemma}
Under Assumption \ref{Assumption} it holds
$$\lim_{\hbar\to 0}\bigg{|} \frac{1}{c(\hbar)}Tr[e^{-\beta H_\hbar}]-\int_{S}d\sigma e^{-\beta h_0(\sigma)}\bigg{|}=0.$$
\end{lemma}
\begin{proof}
On the one hand, since  $e^{-tH_\hbar}$ is trace-class we can use the resolution of the identity of coherent state vectors $\Psi_\hbar^{\sigma}$ \cite[Prop. 6]{Com} and obtain
$$Tr[e^{-\beta H_\hbar}]=c(\hbar)\int_{S}d\sigma\langle \Psi_\hbar^{\sigma}, e^{-\beta H_\hbar}\Psi_\hbar^{\sigma}\rangle.$$ By the spectral theorem, $$\langle \Psi_\hbar^{\sigma}, e^{-\beta H_\hbar}\Psi_\hbar^{\sigma}\rangle=\int_{0}^{\infty}e^{-\beta\lambda}d\nu_\hbar^{\sigma}(\lambda),$$ where
$\nu_\hbar^{\sigma}(F)=\langle P_F^{H_\hbar}\Psi_\hbar^{\sigma},\Psi_\hbar^{\sigma}\rangle$, and $P_F^{H_\hbar}$ denotes the spectral probability measure on $[0,\infty)$ associated to the operator $H_\hbar$. Since the function $x\mapsto e^{-\beta x}$ is convex on $[0,\infty)$ we can apply Jensen's inequality for probability measures, obtaining
$$ e^{-\beta \int_{0}^{\infty}\lambda d\nu_\hbar^{\sigma}(\lambda)}\leq\int_{0}^{\infty}e^{-\beta\lambda}d\nu_\hbar^{\sigma}(\lambda).$$
Since $\Psi_\hbar^{\sigma}\in \mathcal{D}(H_\hbar) \ (\sigma\in S)$ it follows that
$$\int_{0}^{\infty}\lambda d\nu_\hbar^{\sigma}(\lambda)=\langle\Psi_\hbar^{\sigma},H_\hbar\Psi_\hbar^{\sigma}\rangle.$$
Combining the above results  and integrating over the phase space $S$ with respect to the measure $d\mu_\hbar(\sigma)=c(\hbar)d\sigma$ yields the inequality
$$c(\hbar)\int_{S}d\sigma e^{-\beta \check{h}_\hbar(\sigma)}\leq Tr[e^{-\beta H_\hbar}].$$
On the other hand, the second hypothesis  of Assumption \ref{Assumption} and the spectral theorem imply the existence of an orthonormal basis of $\mathcal{H}_\hbar$ given by eigenfuntions $\{\phi_\hbar^{(i)}\}$ of $H_\hbar$. 
Then,
$$\langle \phi_\hbar^{(i)},e^{-\beta H_\hbar}\phi_\hbar^{(i)}\rangle=e^{-\beta\langle \phi_\hbar^{(i)},  H_\hbar\phi_\hbar^{(i)}\rangle}=e^{-\beta c(\hbar)\int_Sd\sigma \hat{h}_\hbar(\sigma)|\Phi^{(i)}_\hbar(\sigma)|^2},$$ where $\Phi_\hbar^{(i)}(\sigma)=\langle\Psi_\hbar^\sigma,\phi^{(i)}_\hbar\rangle$. An application of Jensen's inequality applied to the probability measure $c(\hbar)|\Phi_\hbar^{(i)}(\sigma)|^2d\sigma $ yields
$$\langle \phi_\hbar^{(i)},e^{-\beta H_\hbar}\phi_\hbar^{(i)}\rangle \leq c(\hbar)\int_S e^{-\beta \hat{h}_\hbar(\sigma)}|\Phi^{(i)}_\hbar(\sigma)|^2d\sigma .$$
Taking the sum over all $i$ and observing that $\sum_i|\Phi^{(i)}_\hbar(\sigma)|^2=1$, gives
$$Tr[e^{-\beta H_\hbar}]\leq c(\hbar)\int_S e^{-\beta \hat{h}_\hbar(\sigma)}d\sigma,$$
which exists as a result of hypothesis (iii) of Assumption \ref{Assumption}.
In summary,
$$\int_{S}d\sigma e^{-\beta \check{h}_\hbar(\sigma)}\leq \frac{1}{c(\hbar)}Tr[e^{-\beta H_\hbar}]\leq \int_S e^{-\beta \hat{h}_\hbar(\sigma)}d\sigma.$$
Since $\check{h}_\hbar$ and $\hat{h}_\hbar$ are  both assumed to converge pointwise to $h_0$, the result follows as an application of dominated convergence theorem.
\end{proof}

\begin{proof}[Proof of Proposition \ref{generalcase1}]
The proof is based on the ideas mentioned in the paper \cite{Lieb} and a $C^*$-version of the {\em Peierls-Bogolyubov Inequality} \cite[Thm 7]{Rus}, namely
\begin{align}
\frac{Tr[e^AB]}{Tr[e^A]}\leq \text{log}\bigg{(}\frac{Tr[e^{A+B}]}{Tr[e^A]}\bigg{)},
\end{align}
whenever $B$ is bounded and self-adjoint, $A$ is self-adjoint and bounded above such that $Tr(e^A)< \infty$. For $t>0$ we apply this inequality to $B=-\beta t Q^B_\hbar(f)$ (with $f$ real-valued so that $Q_\hbar^B(f)$ is self-adjoint and compact (cf. Prop. \ref{propertiesberezinmap})), and to $A=-\beta H_\hbar$ which by the hypotheses is self-adjoint, bounded above and its exponential $e^{-\beta H_\hbar}$ has finite trace. Let us define
\begin{align}\label{FHQ}
F_\hbar^Q(t):=-{\beta^{-1}}\text{log}\bigg{[}\frac{1}{c(\hbar)}Tr[e^{-\beta(H_\hbar+t Q_\hbar^B(f))}]\bigg{]}.
\end{align}
Hence, with $t>0$ we see
\begin{align}
[F_\hbar^Q(0)-F_\hbar^Q(-t)]/t\geq \omega_{\hbar}^{\beta}(Q_{\hbar}^B(f))\geq [F_\hbar^Q(t)-F_\hbar^Q(0)]/t.\label{tric1}
\end{align}
An application of the previous lemma yields the inequality
\begin{align*}
\check{Z}_{H_\hbar}\leq Z_{H_\hbar}\leq \hat{Z}_{H_\hbar},
\end{align*}
where $Z_{H_\hbar}$ denotes the quantum partition function associated to $H_\hbar$, i.e. $Z_{H_\hbar}=Tr[e^{-\beta H_\hbar}]$,  $\check{Z}_{H_\hbar}=Tr[Q_\hbar^B(e^{-\beta \check{h}_\hbar})]$, and $\hat{Z}_{H_\hbar}=Tr[Q_\hbar^B(e^{-\beta\hat{h}_\hbar})]$. Note that these expressions make sense by hypothesis (iii) of Assumption \ref{Assumption}. It follows that 
\begin{align*}
\hat{F}^{cl}(0,\hbar)\leq F_\hbar^{Q}(0)\leq \check{F}^{cl}(0,\hbar),
\end{align*}
where, similarly as before, we defined for any $t\in\bR$ the functions
\begin{align*}
&\hat{F}^{cl}(t,\hbar):=-\beta^{-1}\text{log}\bigg{[}\int_{S}e^{-\beta (\hat{h}_\hbar(\sigma)+tf(\sigma))}d\sigma\bigg{]};\\
&\check{F}^{cl}(t,\hbar):=-\beta^{-1}\text{log}\bigg{[}\int_{S}e^{-\beta (\check{h}_\hbar(\sigma)+tf(\sigma))}d\sigma\bigg{]}.
\end{align*}
We point out to the reader that, if $A$ and $B$ are operators, then $(A+B)^*=A^*+B^*$ if  $A$ is densely defined and $B \in \gB(\mathcal{H})$. As a result the operator $H_\hbar+\lambda Q_\hbar^B(f)$ is self adjoint on $\mathcal{D}(H_\hbar)$. Moreover,
as a corollary of \cite[Thm. 4]{Rus} using that $e^{-\beta H_\hbar}$ is trace-class, we observe,
$$Tr[|e^{-\beta(H_\hbar+t Q_\hbar^B(f))}|]=Tr[e^{-\beta(H_\hbar+t Q_\hbar^B(f))}]\leq  Tr[e^{-\beta H_\hbar}e^{-\beta t Q_\hbar^B(f))}] <\infty,$$
so that in particular $e^{-\beta(H_\hbar+ t Q_\hbar^B(f))}$ is trace-class. Repeating the same argument as in the lemma applied to $e^{-\beta(H_\hbar+ t Q_\hbar^B(f))}$ we obtain
\begin{align*}
\int_{S}e^{-\beta (\check{h}_\hbar(\sigma)+t f(\sigma))}d\sigma\leq \frac{1}{c(\hbar)}Tr[e^{-\beta(H_\hbar+t Q_\hbar^B(f))}]\leq \int_{S}e^{-\beta (\hat{h}_\hbar(\sigma)+ t f(\sigma))}d\sigma.
\end{align*}
This combined with \eqref{tric1} yields
\begin{align*}
[\hat{F}^{cl}(t,\hbar)-\check{F}^{cl}(0,\hbar)]/t\leq \omega_{\hbar}^{\beta}(Q_{\hbar}^B(f))\leq [\check{F}^{cl}(0,\hbar)]-\hat{F}^{cl}(-t,\hbar)]/t.
\end{align*}
By the dominated convergence theorem using continuity of the logarithmic function on the positive real axis, we observe 
\begin{align}\label{FHC}
    \check{F}^{cl}(t,\hbar)\to F^{cl}(t,0):=-\beta^{-1}\text{log}\bigg{[}\int_{S}e^{-\beta (h_0(\sigma)+t f(\sigma))}d\sigma\bigg{]}, \ \ (\hbar\to 0),
\end{align}
and similarly, $\hat{F}^{cl}(t,\hbar)\to F^{cl}(t,0)$, ($\hbar\to 0$). We can therefore drop the second zero and we write $F^{cl}(t):=F^{cl}(t,0)$. Hence,
\begin{align*}
[F^{cl}(t)-F^{cl}(0)]/t\leq \limsup_{\hbar\to 0} \omega_{\hbar}^{\beta}(Q_{\hbar}^B(f))\leq [F^{cl}(0)-F^{cl}(-t)]/t.
\end{align*}
Since $F^{cl}$  is differentiable in $t\in\bR$ we must have $\lim_{t\to 0^+}[F^{cl}(t)-F^{cl}(0)]/t=\lim_{t\to 0^+}[F^{cl}(0)-F^{cl}(-t)]/t=\frac{d}{d t}|_{t=0}F^{cl}(t)$. It is not difficult to see that the derivative equals
\begin{align*}
\frac{d}{d t}\bigg{|}_{t=0}F^{cl}(t)=\frac{\int_{S}d\sigma f(\sigma)e^{-\beta h_0(\sigma)}}{\int_{S}d\sigma e^{-\beta h_0(\sigma)}}.
\end{align*}
We conclude that 
\begin{align}\label{classicallimitstatenew}
\lim_{\hbar\to 0} \omega_{\hbar}^{\beta}(Q_{\hbar}^B(f))=\frac{\int_{S}d\sigma f(\sigma)e^{-\beta h_0(\sigma)}}{\int_{S}d\sigma e^{-\beta h_0(\sigma)}}.
\end{align}
This proves the proposition.
\end{proof}

By Proposition \ref{classical Gibbs} it immediately follows that \eqref{classicallimitstatenew} corresponds to the unique classical Gibbs states defined on any Poisson $*$-subalgebra $\tilde{\gA}_0$ of $C_0(S)$ containing $C_c^\infty(S)$.

\subsection{Schr\"{o}dinger operators}\label{SCHrop1}
We now discuss the classical limit of Gibbs states induced by Schr\"{o}dinger operators. The relevant manifold $\mathbb{R}^{2n}$ admits a coherent pure state quantization according (Appendix \ref{Beronrealphasespace}) with ensuing quantization maps defined by compact operators provided the quantized  functions are  taken in $C_0(\bR^{2n})$ \eqref{berquanrealphasespace}.

We consider $\hbar$-dependent\footnote{We stress that under a certain scale separation the physical meaning of the limit in Planck’s constant $\hbar\to 0$ in this context can be interpreted as the limit $m\to\infty$, where $m$ denotes the mass of the quantum particle. Indeed, the limit $\hbar\to 0$ where $\hbar$ occurs as  $-\frac{\hbar^2}{2m}$ in front of the Laplacian at fixed mass (usually set to unity) of a general Schr\"{o}dinger operator can be equivalently obtained by sending $m$ to infinity at fixed $\hbar$ \cite{MorVen2}.} (unbounded) Schr\"{o}dinger operators $H_{\hbar}$ defined on some dense domain of $\mathcal{H}=L^2(\mathbb{R}^n, dx)$. Such operators are defined by
\begin{align}
H_{\hbar}:=\overline{-\hbar^2\Delta + V}, \quad \hbar>0\:,\label{Schroper}
\end{align}
where $\Delta$ denotes the Laplacian on $\mathbb{R}^{n}$, and $V$ denotes multiplication by some real-valued function on $\mathbb{R}^n$, playing the role of the potential. 
In order to meet the hypotheses of Assumption \ref{Assumption} one may require the following conditions on the potential.
\begin{itemize}\label{HypotV}
\item[{ (V1)}] $V$ is a  real-valued smooth function on $\bR^n$.
\item[{ (V2)}]  $\inf_{x\in \bR^n} V(x)= \min_{x\in \bR^n}V(x) = c >-\infty\:. \label{defC}$
\item[{ (V3)}] $e^{-tV} \in L^1(\bR^n)\cap L^\infty(\bR^n)$ for $t>0$.
\end{itemize}
\vspace*{0.2cm}
By standard results it follows that the ensuing Schr\"{o}dinger operator $H_\hbar$ is essentially self-adjoint on $C_0(\mathbb{R}^2)$ and the Schwartz space $\mathscr{S}(\bR^n)$ is included in the domain of $H_\hbar$ \cite{RS2}. 

\begin{example}\label{physicalexamples}
{\em A main example of a potential satisfying the hypotheses \ref{HypotV} is the function $V(q):=(q-1)^2$ on $\bR^n$ describing
\begin{itemize}
\item[(i)] a {\em double well} potential on $\bR^1$, i.e. $n=1$;
\item[(ii)] a {\em Mexican hat} potential on  $\bR^{2}$, with $n=2$.
\end{itemize}
}
\hfill$\blacksquare$
\end{example}
As a result of \cite[Chapter 8]{Sim79}, for $h_0(q,p):=p^2+V(q)$ and $\beta>0$ the bound
\begin{align}\label{boundonoperator}
Tr[e^{-\beta H_\hbar}]\leq \frac{1}{(2\pi\hbar)^n}\int_{\bR^{2n}}d^nqd^npe^{-\beta h_0(q,p)}
\end{align}
automatically holds. As a result of condition $(V3)$, the function $e^{-\beta h_0}\in L^1(\bR^{2n})\cap L^\infty(\bR^{2n})\ \ (\beta>0)$, so that by Theorem \ref{propertiesberezinmap} it follows that the right-hand side of \eqref{boundonoperator} is finite for each $\hbar>0$. Since $e^{-\beta H_\hbar}$ is positive, we now may conclude that $e^{-\beta H_\hbar}$  is a trace-class operator. Therefore, Proposition \ref{quantum Gibbs} applies, so that the state $\omega_\hbar^\beta$ defined by \eqref{KMSq} for the Hamiltonian $H_\hbar$ is the unique $\beta$-KMS state for the one-parameter subgroup $t\mapsto\alpha_t^{H_\hbar}\in \text{Aut}(\gB_\infty(\mathcal{H}))$ generated by $H_\hbar$. In order to show that Proposition \ref{generalcase1} applies it suffices to prove Lemma \ref{berezinlieblemma} for $H_\hbar$. To this avail, since $\Psi_\hbar^{(q,p)}\in \mathscr{S}(\bR^{n})\subset \mathcal{D}(H_\hbar) \ ((q,p)\in\bR^{2n})$ is not difficult to see that (see e.g. \cite{SimCL})
$$\hat{h}_\hbar(q,p)=\langle\Psi_\hbar^{(q,p)},H_\hbar\Psi_\hbar^{(q,p)}\rangle=p^2+\frac{n\hbar}{2}+V(q).$$
Combining the above results  and integrating over the phase space $\bR^{2n}$ with respect to the measure $d\mu_\hbar(q,p)=\frac{1}{(2\pi\hbar)^n}d^nqd^nq$ (for $c(\hbar)=\frac{1}{(2\pi\hbar)^n}$) yields the inequality
$$\frac{1}{(2\pi\hbar)^n}\int_{\mathbb{R}^{2n}}d^nqd^npe^{-\beta (p^2+\frac{n\hbar}{2}+V(q))}\leq Tr[e^{-\beta H_\hbar}]\leq \frac{1}{(2\pi\hbar)^n}\int_{\bR^{2n}}d^npd^nqe^{-\beta h_0(q,p)},$$
or equivalently
$$\int_{\mathbb{R}^{2n}}d^nqd^npe^{-\beta (p^2+\frac{n\hbar}{2}+V(q))}\leq (2\pi\hbar)^nTr[e^{-\beta H_\hbar}]\leq \int_{\bR^{2n}}d^npd^nqe^{-\beta h_0(q,p)}.$$
Since $p^2+\frac{n\hbar}{2}+V(q)$ converges pointwise to $p^2+V(q)$ an application of dominated convergence theorem now proves that
$$\lim_{\hbar\to 0}\bigg{|}(2\pi\hbar)^nTr[e^{-\beta H_\hbar}]-\int_{\mathbb{R}^{2n}}d^nqd^npe^{-\beta h_0(q,p)}\bigg{|}=0,$$ which therefore shows the validity of Lemma \ref{berezinlieblemma}.

\section{Mean-field theories}\label{MFQSS}
Mean-field theories (MFTs) play an important role as approximate models of the more complex nearest neighbor interacting spin systems. Their relatively simple structure allows for a detailed analysis of limit $|\Lambda|\to\infty$, especially in view of spontaneous symmetry breaking (SSB) and phase transitions. 

Homogeneous mean-field quantum spin systems fall into the class of MFTs. They are defined by a single-site Hilbert space $\mathcal{H}_x=\mathcal{H}=\mathbb{C}^k$ and local Hamiltonians of the type
\begin{align}
\tilde{H}_{\Lambda}=|\Lambda|\tilde{h}(T_0^{(\Lambda)},T_1^{(\Lambda)},\cdot\cdot\cdot, T_{k^2-1}^{(\Lambda)}),\label{meanfield}
\end{align}
where $\tilde{h}$ is a polynomial in $k^2-1$ variables, and $\Lambda\subset\mathbb{Z}^d$ denotes a finite lattice on which $\tilde{H}_{\Lambda}$ is defined and $|\Lambda|$ denotes the number of lattice points (see e.g. \cite[Chapter 10]{Lan17}). Here $T_0= 1_{M_k(\mathbb{C})}$, and the matrices $(T_i)_{i=1}^{k^2-1}$ in $M_k(\mathbb{C})$ form a basis of the real vector space of traceless self-adjoint $k\times k$ matrices; the latter may be identified with $i$ times the Lie algebra $\mathfrak{su}(k)$ of $SU(k)$, so that $(T_0,T_1,...,T_{k^2-1})$ is a basis of $i$ times the Lie algebra $\mathfrak{u}(k)$ of the unitary group $U(k)$ on $\mathbb{C}^k$. The macroscopic average spin operators are now defined by
\begin{align}
T_\mu^{(\Lambda)}=\frac{1}{|\Lambda|}\sum_{x\in\Lambda}T_\mu(x), \ \ (\mu=0,\cdots ,k^2-1).
\end{align}
Here $T_\mu(x)$ stands for $I_{k} \otimes \cdots \otimes T_\mu\otimes \cdots \otimes I_{k}$, where $T_\mu$ occupies  slot $x$, and $I_k$ is the unit matrix on $\mathbb{C}^k$.

It follows directly from the definition that such models are characterized by the property that all spins interact with each other which implies that these models are permutation-invariant and that the geometric configuration including the dimension is irrelevant. In what follows we therefore consider homogeneous mean-field quantum spin chains. We will often leave out the term ``homogeneous'' in the forthcoming  discussion.

\begin{example}\label{UnscaledLMG}
{\em 
The Lipkin-Meshkov-Glick Model, or simply the LMG model, used to serve to describe phase transitions in atomic nuclei \cite{Lip}, but later it was found that the LMG model is relevant in the study of many other quantum systems such as cavity quantum electrodynamics (cavity QED) and spontaneous symmetry breaking (SSB) \cite{Morr}. Its Hamiltonian describes a mean-field interaction, and in the one-dimensional case, it is defined on a chain of matrix algebras $M_2(\mathbb{C})^{\otimes N}$ of length $N=|\Lambda|$,\footnote{As homogeneous mean-field quantum Hamiltonians systems are related to the quantization maps  $Q_{1/N}$ defined through \eqref{deformationquantization1}--\eqref{deformationquantization2}, we use the subscript $1/N$ rather than $N$ or $|\Lambda|$.}
\begin{align}
\tilde{H}^{LMG}_{1/N}: &  \underbrace{\mathbb{C}^2 \otimes \cdots  \otimes\mathbb{C}^2}_{N \: times}  \to 
\underbrace{\mathbb{C}^2 \otimes \cdots  \otimes\mathbb{C}^2}_{N \: times};\nonumber \\
\tilde{H}^{LMG}_{1/N} &=\frac{\lambda}{N}\bigg{(}(\sum_{x\in\Lambda}\sigma_1(x))^2+\gamma (\sum_{x\in\Lambda}\sigma_2(x))^2\bigg{)}-B\sum_{x\in\Lambda}\sigma_3(x), \label{CWhamunscaled}
\end{align}
where $\lambda,\gamma,B \in \mathbb{R}$ are given constants defining the type of interaction, the anisotropic in-plane coupling, and the external magnetic field along $z$ direction, respectively.
Regarding \eqref{meanfield} is it not difficult to see that
\begin{align}
\tilde{h}^{LMG}(T_0^{(N)},T_1^{(N)},T_2^{(N)},T_3^{(N)})=\lambda(S_1^2+\gamma S_2^2)-BS_3, \label{tildepol}
\end{align}
with $S_\mu\equiv T_\mu^{(N)}$ ($\mu=1,2,3$).

}
\hfill $\blacksquare$
\end{example}

\subsection{Mean-field quantum spin systems, symbols and strict deformation quantization}

The traditional way of studying the large limit of lattice sites $N=|\Lambda|$ exists under the name {\em thermodynamic limit}, that is, a rigorous formalism in which $N$ as well as the volume $V$ of the system at constant density $N/V$, are sent to infinity. The limiting system constructed in the limit $N=\infty$ is typically identified with quantum statistical mechanics in infinite volume.  In this setting the so-called {\em quasi-local} observables are studied: these give rise to a non-commutative continuous bundles of $C^*$-algebras $\mathcal{A}^{q}$ over base space $I$\footnote{The elements of $I$ can be thought of ``quantize'' values of Planck’s constant $\hbar=1/N$ upon which the limit $N\to\infty$ is formally the same as the limit $\hbar\to 0$.}  defined by
\begin{equation}
I= \{1/N \:|\: N \in  \mathbb{N}_*\} \cup\{0 \}\equiv (1/\mathbb{N}_*) \cup\{0 \}, \label{defI}\
\end{equation}
with the  topology inherited from $[0,1]$ and $\mathbb{N}_*:=1,2,3,\cdots,$. For finite $N$, the fibers $\mathcal{A}_{1/N}$ at $1/N$ are given by the $N$-fold (projective) tensor product of a unital $C^*$-algebra $\gA$ (e.g. a matrix algebra) with itself, and $\mathcal{A}_0$, i.e. the fiber at $N=\infty$ or equivalently at $1/N=0$, is given by the quasi-local algebra  (Appendix \ref{Qss}).

This is not the whole story, since the limit $N\to\infty$ may also describe the relation between classical (spin) theories viewed as limits of quantum statistical mechanics. In this case the  {\em quasi-symmetric} (also called {\em macroscopic}) observables are studied and these induce a continuous bundle of $C^*$-algebras $\mathcal{A}^{c}$  which is defined over the same base space $I=1/\mathbb{N}_*\cup \{0\}$, with fibers at $1/N>0$ given by the $N-$fold symmetric tensor powers of $\gA$ with itself, but as opposed to the quasi-local bundle, the $C^*-$algebra $\mathcal{A}_0^{\pi}$ at $N=\infty$ is {\em commutative} (Appendix \ref{Qss}). Indeed, it can be shown that $\mathcal{A}_0^{\pi}\cong C(S(\gA))$, the $C^*-$algebra of continuous functions over the algebraic state space of the single site algebra $\gA$.

It is precisely the latter interpretation of the limit $N\to\infty$ that relates mean field quantum spin systems to strict deformation quantization, since the  ensuing MF Hamiltonians correspond to quasi-symmetric sequences which in turn are defined by quantization maps, at least when $\gA$ assumes the form $\gA=M_k(\mathbb{C})$  as precisely occurs when dealing with spin systems \cite{LMV} (Appendix \ref{Qss}). To make this precise we need to introduce the notion of a {\bf classical symbol}, that is, generally speaking a function
\begin{align}
\tilde{h}_N:=\sum_{j=0}^{M}N^{-j}\tilde{h}_j+O(N^{-(M+1)}), \ \ \label{classicalsymbol1}
\end{align}
for some $M\in\mathbb{N}$ and where each $\tilde{h}_j$ is a real-valued function on the manifold one considers. The first term $\tilde{h}_0$ is called the {\bf principal symbol}. 

For mean-field theories on a lattice it is the commutative $C^*$-bundle with fibers $S_N(\gA^{N})$ ($N>0$) and $C(S(M_k(\mathbb{C})))$ ($N=\infty$) and  associated quantization maps  $Q_{1/N}$ defined through equations \eqref{deformationquantization1}--\eqref{deformationquantization2} that relates the corresponding quantum Hamiltonian to these symbols. The classical symbol $\tilde{h}_N$ will be a polynomial on $S(M_k(\bC))$ and its image under the maps $Q_{1/N}$ yields the scaled mean-field quantum Hamiltonian $\tilde{H}_{1/N}/N$ in question, i.e. $\tilde{h}_N$ is said to be the {\em classical symbol} of $\tilde{H}_{1/N}$ \cite{Ven2020}. By construction, $\tilde{H}_{1/N}/N$ is a quasi-symmetric sequence in the sense of \eqref{quasisymmnew}. We moreover stress that the scaling factor is essential: it principally says that the mean-field Hamiltonian is of order $N$, an essential property for the study of phase transitions in the regime $N\to\infty$. Furthermore, we will see below that the associated principal symbol exactly plays the role of the polynomial $\tilde{h}$ defined in the very beginning of Section \ref{MFQSS}.

\begin{example}[Example \ref{UnscaledLMG} revisited]
{\em Let us go back to the example of the LMG model.  
We first note that 
 \begin{equation}
 \tilde{H}^{LMG}_{1/N} \in \mathrm{Sym}(M_2(\mathbb{C})^{\otimes N}),
\end{equation}
where $\mathrm{Sym}(M_2(\mathbb{C})^{\otimes N})$ is the range of the symmetrizer introduced in Appendix \ref{Qss}, cf. \eqref{def:SN}.
We scale the quantum LMG model with a global factor $1/N$, and we rewrite $\tilde{H}_{1/N}^{LMG}/N$ as
\begin{align}
&\tilde{H}^{LMG}_{1/N}/N\nonumber=\\& \frac{\lambda}{N^2} \bigg{(}\sum^N_{i \neq j, \:i,j=1} \sigma_1(i)\sigma_1(j)+\gamma \sum^N_{i \neq j, \:i,j=1} \sigma_2(i)\sigma_2(j) \bigg{)} - \frac{B}{N} \sum_{j=1}^N \sigma_3(j) + O(1/N)\nonumber=\\
&\lambda\bigg{(}S_{2,N}(\sigma_1\otimes\sigma_1)+\gamma S_{2,N}(\sigma_2\otimes\sigma_2)\bigg{)}-BS_{1,N}(\sigma_3)+ O(1/N)\nonumber =\\
&Q_{1/N}(\tilde{h}_0^{LMG})+O(1/N),
\label{QNh}
\end{align}
where $O(1/N)$ is meant in norm (i.e. the operator norm on $M_2(\mathbb{C})^{\otimes N}$) and the classical LMG Hamiltonian under the identification $S(M_2(\mathbb{C}))\cong B^3=\{(x,y,z) \ | \ x^2+y^2+z^2\leq 1\}$ reads
\begin{align}
&\tilde{h}_0^{LMG}: B^3\to \mathbb{R};\\
&\tilde{h}_0^{LMG}(x,y,z)=\lambda(x^2+\gamma y^2)-Bz .   
\end{align}
Therefore, the LMG model indeed defines a quasi-symmetric sequence according to \eqref{quasisymmnew}).
}
\hfill $\blacksquare$
\end{example}

\subsection{The free energy in the regime of large particles}\label{The free energy in the regime of large particles}
In the forthcoming discussion  we focus on (local) Gibbs states induced by mean-field quantum spin Hamiltonians $\tilde{H}_{1/N}$. Our interest is the limit $N=|\Lambda|\to\infty$ where, as before, $N=|\Lambda|$ correspond to the number of lattice points.


Let us first recall some basics in the case of matrix algebras.
Each operator $A_N\in \gA^{N}=\bigotimes_{x\in\Lambda}M_{k}(\bC)$ induces 
the Gibbs state defined by \eqref{KMSq}, i.e.  $\omega_{N}^\beta(\cdot)=\frac{Tr(e^{-\beta A_N}\cdot)}{Tr(e^{-\beta A_N})}$. This  is the unique KMS state (at inverse temperature $\beta$) on the matrix algebra $\gA^{N}$ for the Heisenberg dynamics implemented by $A_{N}$. Here, $Tr\equiv Tr_N$ denotes the usual trace on $\mathcal{H}_N=\bigotimes_{x\in\Lambda}\mathcal{H}_x$. 

From a different perspective, one can introduce the local internal energy of any operator $A_N$ by
\begin{align}
    U_{N}(A_N,\omega_{N})=\omega_{N}(A_N),
\end{align}
and the local  entropy  by
\begin{align}
    S_{N}(\omega_{N})=-Tr(\rho_{N}\log(\rho_{N})),
\end{align}
where we have identified $\omega_{N}$ with $\rho_{N}\in \gA^{N}$, i.e. the density matrix $\rho_{N}$ uniquely corresponds to the state $\omega_{N}$ in $S(\gA^{N})$ via $\omega_{N}(\cdot)=Tr(\rho_{N}\cdot)$. Finally, the local free energy at $\beta\in (0,\infty)$ is defined as
\begin{align}
    F_{N}^\beta(A_N,\omega_{N})= U_{N}(A_N,\omega_{N})-\frac{1}{\beta}S_{N}(\omega_{N}),
\end{align}
Given a local Gibbs state $\omega_{N}^\beta(\cdot)=\frac{Tr(e^{-\beta A_N}\cdot)}{Tr(e^{-\beta A_N})}$ on $\mathfrak{A}^{N}$ it can be shown that
\begin{align}\label{finite free energy}
     F_{N}^\beta(A_N,\omega_{N}^\beta)=\inf_{\omega_{N}\in S(\mathfrak{A}^N)}F_{N}^\beta(A_N,\omega_{N})=-\frac{1}{\beta }\log{Tr(e^{-\beta A_N})},
\end{align}
and this infimum is uniquely attained for the local Gibbs state. Hence, the unique $\beta-$KMS local Gibbs state is precisely the unique minimizer of the local free energy. 
We write 
\begin{align}
    F_{N}^\beta(A_N):=\inf_{\omega_{N}\in S(\gA^{N})}F_{N}^\beta(A_N,\omega_{N}).
\end{align}
 The above implies that in finite quantum systems nothing really happens. In contrast, the behavior of the  system in the limit of large particles  (viz. lattice sites) does allow for interesting physical phenomena, in particular when considering mean-field Hamiltonians $\tilde{H}_{1/N}$. To see this, we recall these $\tilde{H}_{1/N}$ are determined (via $\tilde{H}_{1/N}=NH_{1/N}$) by the continuous-cross sections $H=(H_{1/N})_N$ of the continuous bundle of $C^*-$ algebras $\mathcal{A}^c$ (Appendix \ref{Qss}). These sequences $(H_{1/N})_N$ (and thus also $(\tilde{H}_{1/N})_N$) are in particular permutation invariant for each $N>0$, so that the induced Gibbs states are permutation invariant as well, i.e. $\omega_{1/N}^\beta\in S^\pi(\mathcal{A}_{1/N})$ for each $N>0$.\footnote{In what follows we will often use the subscript $1/N$ and notations  $\mathcal{A}_{1/N}=\gA^N$ and  $\mathcal{A}_0=[\gA]^\infty$ defined in  \eqref{Eq: quasi-local bundle} to indicate the correspondence between finite (quantum) and infinite (classical) theory via the notion of continuous bundle of $C^*-$algebras over base space $I$, cf. \eqref{defI}.} By a standard Hahn-Banach argument we can extend the Gibbs states to states $\hat{\omega}_{1/N}^\beta\in S(\mathcal{A}_0)$.\footnote{Despite the fact that each limit point $\omega_0$ of the local Gibbs state is permutation-invariant, i.e. an element of $S^\pi(\mathcal{A}_0)=\{\omega_0\in S(\mathcal{A}_0)\ | \ \omega_0\circ\varphi_M\in S^\pi(\gA^M), \ \forall M\in\mathbb{N}\}$ (see \eqref{def:SN} for the definition), it is not necessarily true that at fixed $N$ the extended state $\hat{\omega}_{1/N}^\beta$ is in $S^\pi(\mathcal{A}_0)$.} In this way, we obtain a sequence  $(\hat{\omega}_{1/N}^\beta)$ of states in $S(\mathcal{A}_0)$ of which we can take a convergent subsequence $\hat{\omega}_{1/{N_j}}^\beta$ with limit $\omega_0^\beta$ (with respect with the weak-$*$ topology) 
\begin{align}\label{limitpointvar}
\omega_0^\beta:=\lim_{j\to\infty}\hat{\omega}_{1/{N_j}}^\beta\,,
\end{align}
which is necessarily permutation invariant.
It follows that for a local sequences $(A_{1/N})_N$ one has
\begin{align}
&\omega_0^\beta([A_{1/N}]_N)=(\omega_0^\beta\circ\varphi^M)(A_{1/M})=\lim_{j\to\infty}(\hat{\omega}_{1/{N_j}}^\beta\circ\varphi^M)(A_{1/M})\nonumber\\&=\lim_{j\to\infty}(\hat{\omega}_{1/{N_j}}^\beta\circ\varphi^{N_j})(\varphi^M_{N_j}A_{1/M})=\lim_{j\to\infty}(\omega_{1/{N_j}}^\beta\circ\varphi^M_{N_j})(A_{1/M}),
\end{align}
where  $\varphi^M(A_{1/M})\to[A_{1/M}\otimes I^{N-M}]_N$ is the canonical  embedding $\mathcal{A}_{1/N}\mapsto \mathcal{A}_0$ (cf. \eqref{localsequence}), and $\omega_{1/{N_j}}^\beta=\hat{\omega}_{1/{N_j}}^\beta\circ\varphi^{N_j}$. By the celebrated quantum De Finetti Theorem the permutation invariant state $\omega_0^\beta$ assumes the form
\begin{align}\label{limiting state1}
     \omega_0^\beta=\int_{S(M_k(\mathbb{C}))}d\mu_{0}^\beta(\omega')(\omega')^{\otimes\infty},
\end{align}
for a unique probability measure $\mu_{0}^\beta$ on $S(M_k(\mathbb{C}))$, where for $\omega'\in S(M_k(\mathbb{C}))$ the state $(\omega')^{\otimes\infty}$ on $\mathcal{A}_0$ denotes the associated product state, which is clearly permutation invariant.
\\\\
In order to extend the notions of free energy to the infinite system, we consider a (not necessarily permutation-invariant) state $\omega$ on the  quasi-local algebra $\mathcal{A}_0$. The restriction of  $\omega$  to $\mathcal{A}_{1/N}$  with respect to the canonical embedding will be denoted by $\omega|_{1/N}$. We furthermore assume the local algebras to be matrix algebras, so that as before
\begin{align}
    \omega_{|_{1/N}}(A_{1/N})=Tr_{N}(\rho_{|_{1/N}}A_{1/N}), \ \ (A_{1/N}\in\mathcal{A}_{1/N})
\end{align}
where $\rho_{|_{1/N}}\in \gB(\mathcal{H}_N)$ is the corresponding density operator. We occasionally adapt the notation $\rho^\omega$ to indicate the dependence on the state $\omega$. For a quasi-symmetric sequence $H=(H_{1/N})_N$ and ensuing Hamiltonian $\tilde{H}_{1/N}:=NH_{1/N}$ we  define the densities
\begin{itemize}
    \item[-] mean-field internal energy $U(\tilde{H},\omega)=\lim_{N\to\infty}\frac{1}{N}\omega_{|_{1/N}}(\tilde{H}_{1/N})$
    \item[-] mean-field entropy $S(\omega)=\lim_{N\to\infty}-\frac{1}{N}tr_N(\rho^\omega_{|_{1/N}}\log{(\rho^\omega_{|_{1/N}})})$
    \item[-] mean-field free energy $F^\beta(\tilde{H},\omega)=U(\tilde{H},\omega)-\frac{1}{\beta}S(\omega)$,\\
\end{itemize}
providing the corresponding weak-$*$ limits exist.
The following result is proved in the case of permutation invariant states.
\begin{proposition}\label{existenceoflimits}
For all permutation invariant states $\omega$ on the quasi-local algebra $\mathcal{A}_0$ and all quasi-symmetric sequences $H=(H_{1/N})_N$, the above limits exist and define weak-$*$ continuous and affine functionals on $S(C(S(M_k(\mathbb{C}))))$.
\end{proposition}
\begin{proof}

Let us first make the following observation. For product states $\omega^{\otimes N}$ and a local sequence $(A_{1/N})_N$ (the quasi-local case follows from this) defined as in \eqref{localsequence}, the following limit
\begin{align}\label{claslilh01}
    \lim_{N\to\infty}\omega^{\otimes N}(A_{1/N})=:a_0(\omega),
\end{align}
exists and defines a continuous function $a_0\in C(S(M_k(\mathbb{C})))$. Since the product state is permutation-invariant, the same function is obtained by considering its symmetrized form $A_{1/N}:=S_N(A_{1/N})$ (again, the result applies to the quasi-symmetric case as well).

By the hypothesis the state $\omega\in S(\mathcal{A}_0)$ is permutation invariant and thus assumes the form \eqref{limiting state1}. The above observation \eqref{claslilh01} implies that, for any quasi-local sequence $(A_{1/N})_N$ one may view $\omega$ as a state on $C(S(M_k(\mathbb{C})))$, in that
\begin{align}\label{limiting state 2}
    \omega(a_0):=\int_{S(M_k(\mathbb{C}))}d\mu_{\omega}(\omega')a_0(\omega').
\end{align}
To prove the theorem we take $H=(H_{1/N})$ to be a symmetric sequence, so that $H_{1/N}=S_{M,N}(a_{1/M})$ for some $M$ whenever $N\geq M$.  Then
\begin{align}\label{meanfieldenergy3}
   U(\tilde{H},\omega)=\lim_{N\to\infty}\frac{1}{N}\omega_{|_{1/N}}(\tilde{H}_{1/N})=\omega_{|_{1/M}}(H_{1/M})=\omega(h_{0}),
\end{align}
where $\omega(h_0)$ defined through \eqref{limiting state 2}.
For the entropy a similar computation \cite{GRV} shows that
\begin{align}\label{entropy3}
    S(\omega)=\lim_{N\to\infty}-\frac{1}{N}tr_N(\rho^\omega_{|_{1/N}}\log{(\rho^\omega_{|_{1/N}})})=-\int_{S(M_k(\mathbb{C}))}d\mu_{\omega}(\omega')Tr(\rho^{\omega'}\log{\rho^{\omega'}}),
\end{align}
with $\rho^{\omega'}\in M_k(\mathbb{C})$ the density matrix associated to the state $\omega'\in S(M_k(\mathbb{C}))$.
In a similar fashion, for each $\omega\in S(M_k(\mathbb{C}))$ we can define $s_0(\omega):=-Tr(\rho^{\omega}\log{\rho^{\omega}})$. Clearly $s_0\in C(S(M_k(\mathbb{C})))$, so that equation \eqref{entropy3} becomes
\begin{align}
    \lim_{N\to\infty}-\frac{1}{N}tr_N(\rho^\omega_{|_{1/N}}\log{(\rho^\omega_{|_{1/N}})})=\omega(s_0).
\end{align}
Hence,
\begin{align}\label{F0new}
    F^{\beta}(\tilde{H},\omega):=\omega(h_0)-\frac{1}{\beta}\omega(s_0), \ \ (\omega\in S(C(S(M_k(\mathbb{C}))))),
\end{align}
or, similarly
\begin{align}\label{MF-limit1}
    F^\beta(\tilde{H},\omega)=\int_{S(M_k(\mathbb{C}))}d\mu_{\omega}(\omega')\bigg{(}h_0(\omega')-\frac{1}{\beta}s_0(\omega')\bigg{)}.
\end{align}
The fact that the limits define weak-$*$ continuous and affine functionals on $S^\pi(\mathcal{A}_0)$ is standard, as explained e.g. in \cite[Prop 3.2]{GRV}. To conclude, we  observe that $S^\pi(\mathcal{A}_0)$, the convex set of permutation-invariant states on the algebra of equivalence classes of quasi-local sequences, is isomorphic to $S(\mathcal{A}_0^\pi)$, i.e. the state space of the algebra of equivalence classes of quasi-symmetric sequences. As explained in Appendix \ref{Qss}, the algebra $\mathcal{A}_0^\pi$ is isomorphic to $C(S(M_k(\mathbb{C})))$. This completes the proof.
\end{proof}
We have the following result for the mean-field free energy applied to Gibbs states.
\begin{theorem}\label{Gibbsstatelimit}
Let $\beta$ be a fixed inverse temperature and let $H=(H_{1/N})_N$ be a continuous cross-section of the continuous bundle of $C^*-$algebras defined by $\mathcal{A}^{c}$ (cf. Appendix \ref{Qss}). Writing $\tilde{H}_{1/N}=NH_{1/N}$, one has 
\begin{align}
     \frac{1}{N}F_{1/N}^\beta(\tilde{H}_{1/N})\to F^{\beta}(\tilde{H}) \ \text{when} \ N\to\infty,
\end{align}
where
\begin{align}\label{twoequalities}
    F^{\beta}(\tilde{H})=\inf_{\omega\in S(C(S(M_k(\mathbb{C})))}F^{\beta}(\tilde{H},\omega)=\inf_{\rho\in S(M_k(\mathbb{C}))}F^{\beta}(\tilde{H},\delta_\rho),
\end{align}
the function $F^{\beta}(\tilde{H},\cdot)$ defined by \eqref{F0new} and $\delta_\rho$ denotes the state given by point evaluation (i.e. $\delta_\rho(f)=f(\rho))$.
\end{theorem}
\begin{remark}
{\em In view of Theorem \ref{Gibbsstatelimit} we can extract a subsequence of the Gibbs state converging to a probability measure supported on the set of minizers of $F^{\beta}(\tilde{H},\cdot)$. In the case of a unique minimizer we obtain convergence for the whole sequence.
$\blacksquare$\hfill
}
\end{remark}
\begin{proof}[Proof of Proposition \ref{Gibbsstatelimit}]
Let us first prove the first equality in \eqref{twoequalities}. For simplicity we take the continuous-cross section $(H_{1/N})_N$ to be symmetric (for the quasi-symmetric case the result does not change due to the uniform topology in which these sequences are approximated by a symmetric one).
An upper bound for the mean-field free energy is obtained as follows. Given $\omega\in S(M_k(\mathbb{C}))$ we clearly have
\begin{align}
    \frac{1}{N}F_{1/N}^\beta(\tilde{H}_{1/N})\leq \frac{1}{N}F_{1/N}^\beta(\tilde{H}_{1/N},\omega^{\otimes N}).
\end{align}
Applying the limsup to the above inequality yields
\begin{align}\label{inequalitymeanfield}
    \limsup_{N\to\infty}\frac{1}{N}F_{1/N}^\beta(\tilde{H}_{1/N})\leq F^{\beta}(\tilde{H},\omega^{\otimes\infty}),
\end{align}
where we have used Proposition \ref{existenceoflimits} applied to the permutation-invariant state $\omega^{\otimes \infty}$. In particular,
\begin{align}\label{boundupper}
    \limsup_{N\to\infty}\frac{1}{N}F_{1/N}^\beta(\tilde{H}_{1/N})\leq F^{\beta}(\tilde{H}).
\end{align} 
For the lower bound we proceed in a similar way as in Proposition \ref{existenceoflimits} considering a fixed convergent subsequence  $\omega_{1/{N_j}}^\beta\subset\omega_{1/{N}}^\beta$ of the ensuing Gibbs state induced by $\tilde{H}_{1/N}$ as explained in the text around \eqref{limitpointvar}.
It follows that
\begin{align}\label{internal energy inequality}
&\lim_{j\to\infty}\frac{1}{N_j}\omega_{1/N_j}^\beta(\tilde{H}_{1/N_j})=\int_{S(M_k(\mathbb{C}))}d\mu_{0}^\beta(\omega')h_0(\omega')=\omega_0^\beta(h_0),
\end{align}
where $\omega_0^\beta$ is the relevant limit point and $h_0$ is defined through $\eqref{claslilh01}$. 
In order to estimate the (minus) entropy term we use the sub-additivity property 
\begin{align}\label{inequality for entropy}
   Tr[\rho_{1/{N_j}}^\beta\log{\rho_{1/{N_j}}^\beta}]\geq\floor*{\frac{N_j}{n}}Tr[\rho_{1/{N_j}}^{(n)}\log{\rho_{1/{N_j}}^{(n)}}] + Tr[\rho_{1/{N_j}}^{(N_j-n\floor*{\frac{N_j}{n}})}\log{\rho_{1/N_j}^{(N_j-n\floor*{\frac{N_j}{n}})}}],
\end{align}
with $\floor{x}$ the floor function and $\rho_{1/N_j}^{(n)}$ is the $n^{th}-$particle reduced density matrix associated with $\tilde{H}_{1/N_j}$.\footnote{Here we have omitted the very unreadable superscript $\beta$ in $\rho_{1/{N_j}}^{\beta,(n)}$.} For any density matrix $\sigma\in M_k(\mathbb{C})$, we can estimate the second term in \eqref{inequality for entropy} in the following manner
\begin{align}\label{secondestimate}
    &Tr\bigg{[}\rho_{1/N_j}^{(N_j-n\floor*{\frac{N_j}{n}})}\log{\rho_{1/N_j}^{(N_j-n\floor*{\frac{N_j}{n}})}}\bigg{]}\nonumber\\&=Tr\bigg{[}\rho_{1/N_j}^{(N_j-n\floor*{\frac{N_j}{n}})}\bigg{(}\log{\rho_{1/N_j}^{(N_j-n\floor*{\frac{N_j}{n})}}}-\log{\sigma^{\otimes (N_j-n\floor*{\frac{N_j}{n}})}\bigg{)}}\bigg{]} \nonumber\\&+
    Tr\bigg{[}\rho_{1/N_j}^{(N_j-n\floor*{\frac{N_j}{n}})}\log{\sigma^{\otimes (N_j-n\floor*{\frac{N_j}{n}})}}\bigg{]}\nonumber\\&\geq
    \bigg{(}N_j-n\floor*{\frac{N_j}{n}}\bigg{)}Tr\bigg{[}\rho_{1/N_j}^{(1)}\log{\sigma}\bigg{]},
\end{align}
using the basic properties of the partial trace and the fact that for any density matrices $\sigma_1,\sigma_2$ 
\begin{align}
    Tr[\sigma_1(\log{\sigma_1}-\log{\sigma_2})]\geq 0,
\end{align}
as a result of Klein's inequality.
If we take $\sigma$ to be a density matrix of the form $e^{a}/Tr[e^{a}]$ where $a\geq 0$, it follows that $Tr\bigg{[}\rho_{1/N_j}^{(1)}\log{\sigma}\bigg{]}$ is bounded from below independently of $N$. Consequently,
\begin{align}\label{entropy inequality}
     \liminf_{j\to\infty}\frac{1}{N_j}Tr[\rho_{1/N_j}^\beta\log{\rho_{1/N_j}^\beta}]\geq\frac{1}{n}Tr[\rho_0^{(n)}\log{\rho_0^{(n)}}],
\end{align}
for all $n\in\mathbb{N}$. By a similar argument as in the proof in Proposition \ref{existenceoflimits} we obtain  
\begin{align}\label{entropy inequality1}
    \liminf_{n\to\infty}-\frac{1}{n}Tr[\rho_0^{(n)}\log{\rho_0^{(n)}}]=\int_{S(M_k(\mathbb{C}))}Tr(\rho^{\omega'}\log{\rho^{\omega'}})d\mu_{0}^\beta(\omega')=\omega_0^\beta(s_0),
\end{align}
where $s_0$ is defined in Proposition \ref{existenceoflimits}.
In view of \eqref{internal energy inequality}--\eqref{entropy inequality1} it now follows that
\begin{align*}
    &\liminf_{j\to\infty}\frac{1}{N_j}F_{1/N_j}^\beta(\tilde{H}_{1/N_j})\geq
    \int_{S(M_k(\mathbb{C}))}d\mu_{0}^\beta(\omega')\bigg{(}h_0(\omega')+\frac{1}{\beta}Tr(\rho^{\omega'}\log{\rho^{\omega'}})\bigg{)}=F^{\beta}(\tilde{H},\omega_0^\beta). 
\end{align*}
This and the already obtained upper bound \eqref{boundupper} in particular show that
\begin{align}\label{finalthing}
    \lim_{j\to\infty}\frac{1}{N_j}F_{1/N_j}^\beta(\tilde{H}_{1/N_j})=F^{\beta}(\tilde{H}).
\end{align}
Moreover, it may be clear that \eqref{finalthing} holds for any convergent subsequence. Hence,
$$\lim_{N\to\infty}\frac{1}{N}F_{1/N}^\beta(\tilde{H}_{1/N},\omega_{1/N}^\beta)=F^{\beta}(\tilde{H}).$$
This proves the first equality of \eqref{twoequalities}. 

For the second equality in \eqref{twoequalities} we observe that the infimum of  $F^{\beta}(\tilde{H},\cdot)$ is attained in the extreme boundary of $S(C(S(M_k(\mathbb{C}))))$ as a result of Bauer minimum principle, which is applicable since $F^{\beta}(\tilde{H},\cdot)$ is weak-$*$ continuous and affine. The proof is concluded by observing that $\partial_e S(C(S(M_k(\mathbb{C}))))\cong S(M_k(\mathbb{C}))$, under the isomorphism $\delta_\rho(f)=f(\rho)$.
\end{proof}
\begin{remark}
{\em 
The mean-field free energy functional, hence well-defined for all permutation invariant states on the quasi-local algebra, can therefore be recasted as functional on $S(M_k(\mathbb{C}))$. Since $S(M_k(\mathbb{C}))$  admits a Poisson structure \cite{LMV} this perfectly resembles the idea of the classical limit introduced in $\S$\ref{par:claslim}.
$\blacksquare$\hfill
}
\end{remark}

To conclude this section we stress that the minimizers $\rho$ of the mean-field free energy functional obtained by Theorem \ref{Gibbsstatelimit} are solutions of the so-called {\em Gap-equation} \cite{NDRW}
\begin{align}
    \rho(A)=\frac{Tr(e^{-\beta \tilde{H}_{eff}(\rho)}A)}{Tr(e^{-\beta \tilde{H}_{eff}(\rho)})}, \ \ (A\in M_k(\mathbb{C})),
\end{align}
where $\tilde{H}_{eff}$ is the effective one-particle Hamiltonian which depends on $\rho$. Hence, each minimizer corresponds to a Gibbs state. Of course, each of them satisfies the KMS condition on $M_k(\mathbb{C})$, but for an a priori different $\rho-$dependent dynamics.\footnote{This relies on the fact that due to the long-range interactions inherent to mean-field models, no global (uniform) dynamics can be defined.} This is somewhat unsatisfactory in the algebraic formulation of phase transitions for these models where one usually wishes to investigate the presence of multiple KMS states for a {\em given} dynamics. A correct characterization of a ``classical'' KMS condition on $S(M_k(\mathbb{C}))$  therefore remains to be settled.

\subsection{The classical limit of Gibbs states induced by mean-field theories}\label{The classical of Gibbs states induced by mean field theories}

We discuss the existence of the classical limit of the Gibbs state associated to a general mean-field quantum spin Hamiltonian. As opposed to Section \ref{claslimsymplecticmanifolds} where the classical limit corresponded to the unique Gibbs state satisfying the classical KMS condition (Proposition \ref{generalcase1}), the classical limit of the Gibbs state associated with a mean-field Hamiltonian may not even exist.
To illustrate this we make the following observations. Consider two continuous cross-sections $H=(H_{1/N})_N$ and $A=(A_{1/N})_N$ of the continuous $C^*-$bundle $\mathcal{A}^c$ (cf. Appendix \ref{Qss}), fix a real parameter $t$ and consider the sum $H+tA$. As seen before the following mean-field free energy 
\begin{align*}
    F^\beta(\tilde{H},\omega,t):=U(\tilde{H}+t\tilde{A},\omega)-\frac{1}{\beta}S(\omega)=U(\tilde{H},\omega)+tU(\tilde{A},\omega)-\frac{1}{\beta}S(\omega),
\end{align*}
exists for permutation invariant states as a result of Proposition \ref{existenceoflimits}. For any fixed $t\in\bR$ it can be shown that the set of minimizers of $F^\beta(\tilde{H},\cdot,t)$ is non-empty.

We now recall the following fact. For  a concave, continuous function $G$  over a normed space $X$, an element $r_x\in X^*$ is said to be a tangent functional to the graph of $G$ at $x$ if
\begin{align}
G(x + \xi)\leq G(x) + r_x(\xi)
\end{align}
for all $\xi\in X$. Moreover, there is a unique tangent functional to the graph of $G$ at $x$, if and only if $G$ is differentiable at $x$ and in this case the tangent functional is the
derivative of $G$ at $x$. Thus, differentiability  of $G$  at $x$ is equivalent to uniqueness of  the  tangent functional at $x$.

The idea is to apply this to the mean-field free energy functional $F^\beta(\tilde{H},t):= \inf_{\omega}\{F^\beta(\tilde{H},\omega,t)$. We have the following result.
\begin{lemma}
The function $t\mapsto F^\beta(\tilde{H},t):=\inf_{\omega}\{F^\beta(\tilde{H},\omega,t)\}$ is concave and  continuous.
\end{lemma}
\begin{proof}
These are mainly standard arguments based on the affinity of the map $t\mapsto F^\beta(\tilde{H},\omega,t)$. Let $R(t)=F^\beta(\tilde{H},t)$.   Given $\alpha\in[0,1]$ and $t,s\in\bR$ we may thus estimate
\begin{align}
&R(\alpha t+(1-\alpha)s)\geq
\alpha \inf_{\omega}\{F^\beta(\tilde{H},\omega,t)\}+(1-\alpha) \inf_{\omega}\{F^\beta(\tilde{H},\omega,s)\}\nonumber\\&=\alpha R(t)+ (1-\alpha)R(s).
\end{align}
Hence, $R$ is concave. Continuity of $R(t)$ again follows from standard arguments and the fact that $S(M_k(\mathbb{C}))$ is compact.
\end{proof}
In view the previous discussion differentiability of $R$ would be equivalent to uniqueness of the tangent functional associated to the graph of $R$ at $t$. We see below that this condition is sufficient for the existence of the classical limit of the Gibbs state.

To this avail, we let $H=(H_{1/N})_N$ be a continuous cross-section of the $C^*-$bundle $\mathcal{A}^{c}$, and consider the Gibbs state induced by the local Hamiltonians $\tilde{H}_{1/N}=NH_{1/N}$. The idea is to evaluate this state on those $A_{1/N}$ for which $A=(A_{1/N})_N$ is a continuous cross-section of the same bundle $\mathcal{A}^{c}$.
In view of the proof of Proposition \ref{generalcase1} for a real parameter $t\in \bR$ we obtain the following inequalities
\begin{align}\label{Ruskaiineq2}
&\frac{\log{(Tr[e^{-\beta(\tilde{H}_{1/N})}])} -\log{(Tr[e^{-\beta(\tilde{H}_{1/N}-t\tilde{A}_{1/N})}])}}{\beta N t}\geq\omega_{1/N}^\beta(A_{1/N})\nonumber\\&\geq \frac{\log{(Tr[e^{-\beta(\tilde{H}_{1/N}+t\tilde{A}_{1/N})}])} -\log{(Tr[e^{-\beta(\tilde{H}_{1/N})}])}}{\beta N t},
\end{align}
where $\tilde{A}_{1/N}=NA_{1/N}$. Analogous to \eqref{finite free energy}, for $t>0$ the local free energy associated to $\tilde{H}_{1/N}+t\tilde{A}_{1/N}$ assumes the form
\begin{align}
F_{1/N}^\beta(\tilde{H}_{1/N},t)=-\frac{1}{\beta}\log{(Tr[e^{-\beta(\tilde{H}_{1/N}+t\tilde{A}_{1/N})}])},
\end{align}
so that \eqref{Ruskaiineq2} reads
\begin{align*}
&\frac{\frac{1}{N}F_{1/N}^\beta(\tilde{H}_{1/N},0)-\frac{1}{N}F_{1/N}^\beta(\tilde{H}_{1/N},-t)}{t}\geq \omega_{1/N}^\beta(A_{1/N})\geq\nonumber\\& \frac{\frac{1}{N}F_{1/N}^\beta(\tilde{H}_{1/N},t)-\frac{1}{N}F_{1/N}^\beta(\tilde{H}_{1/N},0)}{t}.
\end{align*}
By Theorem \ref{Gibbsstatelimit} it follows that
\begin{align*}
&\limsup_{N\to\infty}\omega_{1/N}^\beta(A_{1/N})\leq \frac{F^\beta(\tilde{H},0)-F^\beta(\tilde{H},-t)}{t};\\ 
&\liminf_{N\to\infty}\omega_{1/N}^\beta(A_{1/N})\geq\frac{F^\beta(\tilde{H},t)-F^\beta(\tilde{H},0)}{t}.
\end{align*}
Moreover, since $F^\beta(\tilde{H},t)$ is concave in $t$ the left and right derivative both exist
\begin{align*}
&F_-^\beta:=\lim_{t\to 0^+}\frac{F^\beta(\tilde{H},0)-F^\beta(\tilde{H},-t)}{t};\\
&F_+^\beta:=\lim_{t\to 0^+}\frac{F^\beta(\tilde{H},t)-F(\tilde{H},0)}{t}.
\end{align*}
If they are equal the function $F^\beta(\tilde{H},\cdot)$ is differentiable at $t=0$ and one finds $\frac{d}{dt}|_{t=0}F^\beta(\tilde{H},t)=\lim_{N\to\infty}\omega_{1/N}^\beta(A_{1/N})$ (cf. Thm. \ref{generalcase1}), which for macroscopic observables is defined by \eqref{limiting state 2}.\footnote{Notice that the same statement is true for quasi-local sequences $A=(A_{1/N})_N$, as the local Gibbs state is permutation-invariant for each $N$, and so are its limit points.} 
Of course, if the the minimizer for the mean-field free energy is unique then it is clear that the classical limit exists, but this in general does not hold. 
In any case, the following result is automatic due to the above comments. 

\begin{corollary}
The classical limit of the sequence of local Gibbs states exists (cf. \eqref{claslim1}) if the mean-field free energy $\mathbb{R}\ni t\mapsto F^\beta(\tilde{H},t)$ is differentiable at $t=0$.
\hfill $\blacksquare$
\end{corollary}
In more complex physical situations, the classical limit typically does not exist. Nevertheless, we have seen in Theorem \ref{Gibbsstatelimit} that in such a cases one  is still allowed to study the possible weak-$*$ limit points of the sequence of Gibbs states.

\begin{remark}
{\em 
We would like to point out to the reader that the main difference with the proof of Prop. \ref{quantum Gibbs} is a different scaling in the free energy, which is due to the fact that a one-particle system is considered there. In that case, the analog of the limiting free energy is just the function \eqref{FHC}, which is obviously differentiable and therefore ensures the existence of the classical limit of the relevant Gibbs state. For mean-field quantum spin systems, the scaling is different due to the many-body character; the pertinent mean-field free energy is not necessarily differentiable anymore as a function of $t$. Consequently, the limiting Gibbs state may not exist. In addition, no coherent state are available on a Poisson manifold. Therefore, the conditions of Assumption \ref{Assumption} do  not apply to mean-field systems, despite that fact both the upper and lower symbols will still  converge to the principal symbol, which is always a polynomial in $k^2-1$ real variables.
\hfill $\blacksquare$
}
\end{remark}

\subsection{Mean-field models with symmetry}\label{symmtric case}
In this section we examine another special class of local Gibbs states and study the possible convergence to some probability measure on $S(M_k(\mathbb{C}))$. More specifically, we consider mean-field theories for which there is an additional symmetry implemented by a Lie group $G$ (assumed to be compact or discrete). In other words, we assume the existence of a group action on the manifold $S(M_k(\mathbb{C}))$. For any $g\in G$ we denote the pullback of the action of $G$ on functions $f:S(M_k(\mathbb{C}))\to\mathbb{C}$ by $\zeta_g$, i.e
\begin{align}\label{inducedGaction}
   \zeta_g(f)(x_1,\cdots,x_{k^2-1})=f(g^{-1}\cdot(x_1,\cdots,x_{k^2-1})).
\end{align}
Now, if the associated principal symbol $\tilde{h}_0$ is invariant under the induced $G$ action defined by \eqref{inducedGaction}, it can be shown that there is a unitary representation $g\mapsto U_g$ such that 
\begin{align}
    \tilde{h}_0(T_0,T_1,\cdot\cdot\cdot, T_{k^2-1})=\tilde{h}_0(T_0,U_gT_1U_g^*,\cdot\cdot\cdot, U_gT_{k^2-1}U_g^*),
\end{align}
and hence the sequence of local mean-field Hamiltonians is $G$-invariant, and as seen before, this transfers to the ensuing local Gibbs state and  each possible limit point.  We have the following result on Gibbs states which once again encompasses the interplay between the quantization map and the classical limit.

\begin{proposition}
Let $\tilde{H}_{1/N}$ define a homogeneous mean-field theory with compact or discrete Lie group $G$ such that its principal symbol $\tilde{h}_0$ is $G-$invariant under the map \eqref{inducedGaction}. Then each limit point of the ensuing sequence $(\omega_{1/N}^{\beta})_N$ of local Gibbs states gives a $G$-invariant probability measure  $\mu_0^{\beta}$ whose support is concentrated on some $G$-orbit in $S(M_k(\mathbb{C}))$. If there is only one such orbit, the classical limit exists, in that
\begin{align}
\lim_{N\to\infty}\bigg{|}\omega_{1/N}^\beta(Q_{1/N}(f))-\omega_0^{\beta}(f)\bigg{|}=0, \ \ (f\in C(S(M_k(\mathbb{C}))),
\end{align}
where $\omega_0^{\beta}\in S(C(S(M_k(\mathbb{C}))))$ is identified with the unique $G$-invariant probability measure   $\mu_0^{\beta}$ and $Q_{1/N}$ denotes the quantization map defined through \eqref{deformationquantization1}--\eqref{deformationquantization2}.
\end{proposition}
The proof is a direct consequence of \cite[Prop 2.9]{GRV} combined with the previous results and is therefore omitted.

\subsection{The classical limit in the limit of large spin quantum number}\label{spinnumbernew}
In the final part we discuss the classical limit for local Gibbs states in the regime of large spin quantum number. To this avail the relevant manifold is the two-sphere $S^2$ which can be shown to admit a coherent pure state quantization (cf. Appendix \ref{cpsqS}) with ensuing Berezin quantization maps defined by \eqref{defquan3}, satisfying the conditions of a strict deformation quantization in the sense of Definition \ref{def:deformationq} as well. The pertinent Hilbert space $\mathcal{H}_N$ has dimension $N+1$. 
Due to the uniform boundedness in $N$ the limit of large number of particles is physically not that interesting. In the same spirit as for mean-field theories the idea is to scale the sections with a factor $N$ and to consider the ensuing operators $\tilde{H}_{1/N}:=NH_{1/N}$. To see what these physically represent we make the following observation.

The Hilbert space $\mathcal{H}_N$ can be realized as the symmetrized tensor product of $\bigotimes^N\mathbb{C}^2$ used to describe permutation invariant systems of $N$ qubits. This is however not the only interpretation as already noticed by E. Majorana \cite{Mar}. Indeed, he observed that a permutation invariant system of $N$ qubits can be identified with an arbitrary pure quantum state of dimension $2J$, where $J$ is the {\bf spin quantum number} defined by $J:=N/2$. Therefore, the limit $N\to\infty$ may also be interpreted as a classical limit in the spin quantum number $J$ of a single quantum spin system. It is the latter understanding which gives physical significance to the operators $\tilde{H}_{1/N}$. 
Indeed, by a result obtained by Lieb \cite{Lieb} one has a correspondence between operators $A$ on $\mathbb{C}^{N+1}$ and continuous functions $f\in C(S^2)$ (upper symbols) such that $A=Q_{1/N}^B(f)$. In the particular case where $A$ denotes a quantum spin operator, the functions $f$ can be determined (see Table \ref{TABLE} below).
\begin{table}[ht]
\begin{center}
\begin{tabular}
{ |p{3cm}||p{9cm}|p{3cm}|p{3cm}|}
 \hline
 $\text{Spin operator}$ & $f(\theta,\phi)$ \\
 \hline
 $S_3$   &  $\frac{1}{2}(N+2)\cos{(\theta)}$\\
 $S_1$ & $\frac{1}{2}(N+2)\sin{(\theta)}\cos{(\phi)}$ \\
 $S_2$ & $\frac{1}{2}(N+2)\sin{(\theta)}\sin{(\phi)}$ \\
 \hline
\end{tabular}
\caption{Quantum spin operators on $\mathbb{C}^{N+1}$ and their corresponding upper symbols $G$ in spherical coordinates ($\theta\in (0,\pi)$, $\phi\in (0,2\pi)$).}
\label{TABLE}
\end{center}
\end{table}
From this table one directly observes that, analogous to the mean-field Hamiltonians, the spin operators $S_\mu$ ($\mu=1,2,3$) become unbounded as $N\to\infty$. It is furthermore clear that the scaled operators $\frac{1}{N}S_{\mu}$ are asymptotically norm-equivalent to the operators $Q_{1/N}^B(x_\mu)$, where the $x_\mu$ denote the standard coordinates on the sphere. Inspired by this idea we obtain the following result on the classical limit in the regime of large spin quantum numbers.

\begin{proposition}\label{spinnumber}
Let $H=(H_{1/N})_N$ be a continuous cross-section of the continuous bundle of $C^*-$algebras defined in Appendix \ref{cpsqS}. Consider $\tilde{H}_{1/N}=NH_{1/N}$. Then, there is a  $h_0\in C(S^2)$ such that 
\begin{align}
     \frac{1}{N}F_{1/N}^\beta(\tilde{H}_{1/N})\to\inf_{\Omega\in S^2}{h_0(\Omega)}, \ \text{when} \ N\to\infty.
\end{align}

\end{proposition}

\begin{proof}
By definition of the continuous bundle there is a $h_0\in C(S^2)$ such that $H_{1/N}$ satisfies
\begin{align}\label{brbrbr}
    ||H_{1/N}-Q_{1/N}^B(h_0)||_N\to 0, \ \text{when} \ N\to\infty,
\end{align}
where $Q_{1/N}^B$ is defined by \eqref{defquan3}. To prove the proposition, a simple estimation shows that it suffices to prove the statement for  $NQ_{1/N}^B(h_0)$. 
Working with this operator we first use Berezin-Lieb's inequality and obtain
\begin{align*}
\frac{(N+1)}{4\pi}\int_{S^2}d\mu_L(\Omega)e^{-N\beta \langle \Psi_N^{\Omega},Q_{1/N}^B(h_0)\Psi_N^{\Omega}\rangle}\leq Tr[e^{-\beta N Q_{1/N}^B(h_0)}]\leq Tr[Q_{1/N}^B(e^{-\beta N h_0})].
\end{align*}
This implies that
 \begin{align}
    &-\frac{1}{N\beta}\log{\bigg{(}\frac{N+1}{4\pi}\int_{S^2}e^{-\beta N h_0(\Omega)}d\mu_L(\Omega)\bigg{)}}\leq \frac{1}{N}F_{1/N}^\beta(NQ_{1/N}^B(h_0))
    \nonumber\\&\leq
-\frac{1}{N\beta}\log{\bigg{(}\frac{N+1}{4\pi}\int_{S^2}d\mu_L(\Omega)e^{-N\beta \langle \Psi_N^{\Omega},Q_{1/N}^B(h_0)\Psi_N^{\Omega}\rangle}\bigg{)}}    
    \label{estimatenew}
\end{align}
where we have applied \eqref{niceformula}
Since $\Omega\mapsto \langle \Psi_N^{\Omega}, Q_{1/N}^B(h_0)\Psi_N^{\Omega}\rangle$ converges  to $\Omega\mapsto h_0(\Omega)$ uniformly,\footnote{This is not difficult to see using the explicit form of the coherent spin states defined through \eqref{om1}--\eqref{om2}.} it follows from Varadhan's Theorem that $\frac{1}{N}F_{1/N}^\beta(NQ_{1/N}^B(h_0))$ converges to $\inf_{\Omega\in S^2}{h_0(\Omega)}$. 
This concludes the proof of the proposition.
\end{proof}

\appendix

\section{Canonical examples of coherent pure state quantizations}\label{pssm}

\subsection{Coherent pure state quantization of $\mathbb{R}^{2n}$}\label{Beronrealphasespace}
We consider the manifold $\mathbb{R}^{2n}$ with symplectic structure $\sum_{k=1}^n dp_k \wedge dq^k$. 
It is well-know that this manifold admits a strict deformation quantization in the sense of Definition \ref{def:deformationq}. Additionally, it turns out that this manifold admits a coherent pure state quantization as well. We hereto recall a result in \cite[II. Prop 2.3.1]{Lan98}.
\begin{theorem}[II. Prop. 2.3.1]
Let $I=[0,\infty)$ and $\mathcal{H}_\hbar=L^2(\mathbb{R}^n,dx)$ for each $\hbar>0$.  Denote by $\mu_L$ the Liouville measure on $\bR^{2n}$ which coincides with the standard $2n$-dimensional Lebesgue measure $d^nqd^np$. For any $(p,q)\in\bR^{2n}$ define a unit vector $\Psi_\hbar^{(q,p)}\in \mathcal{H}_\hbar$ by 
\begin{align}
\Psi_\hbar^{(q,p)}(x) := (\pi \hbar)^{-n/4}e^{-\frac{i}{2}p\cdot q/\hbar} e^{ip\cdot x/\hbar} e^{-(x-q)^2/2\hbar} \:, \quad (x \in \bR^n)\:, \hbar >0.\:
\end{align}
Denote the projection of $\Psi_\hbar^{(q,p)}\in\mathbb{S}\mathcal{H}$ to $\mathbb{P}\mathcal{H}$ by $\psi_\hbar^{(q,p)}$. Then the choices, 
\begin{align}
&q_\hbar(p,q)=\psi_\hbar^{(q,p)};\\
& c(\hbar)=\frac{1}{(2\pi\hbar)^n},
\end{align}
so that $d\mu_\hbar(p,q)=\frac{d^npd^nq}{(2\pi\hbar)^n}$, yield a coherent pure state quantization of $(\mathbb{R}^{2n},\sum_{k=1}^n dp_k \wedge dq^k)$.  The unit vectors $\Psi_\hbar^{(q,p)}$ are  called  {\bf Schr\"{o}dinger coherent states}.  For $f\in L^\infty(\bR^{2n})$ the Berezin quantization map $Q_\hbar^B(f)$ then reads
\begin{align}\label{berquanrealphasespace}
Q_\hbar^B(f)=\frac{1}{(2\pi\hbar)^n}\int_{\bR^{2n}}d^nqd^np f(q,p)|\Psi_\hbar^{(q,p)}\rangle\langle \Psi_\hbar^{(q,p)}|,
\end{align} 
where $|\Psi_\hbar^{(q,p)}\rangle\langle\Psi_\hbar^{(q,p)}|$ is the one-dimensional projection onto the subspace spanned by $\Psi_\hbar^{(q,p)}\in \mathcal{H}_\hbar$.
\hfill $\blacksquare$
\end{theorem}

Furthermore, for each $\hbar>0$ the map $Q_\hbar^B:C_0(\bR^2)\to\gB_\infty(L^2(\bR^n))$ is a surjection. 

\subsection{Coherent pure state quantization of $S^2$}\label{cpsqS}
In this section we focus on the symplectic manifold $(S^2,\sin{\theta}d\theta\wedge d \phi)$ where $\theta\in (0,\pi)$ and $\phi\in (0,2\pi)$. The existence of a strict deformation quantization has been shown in \cite[Thm. 8.1]{Lan17}. Also now a coherent pure state quantization can be constructed. For $f\in C(S^2)$ the Berezin quantization map is defined by
\begin{align}\label{defquan3}
Q_{1/N}^B(f)& :=
 \frac{N+1}{4\pi}\int_{S^2}d\Omega f(\Omega)|\Psi_N^\Omega\rangle\langle\Psi_N^\Omega|\: \in \gB(\text{Sym}^N(\mathbb{C}^2)),
\end{align}
where $d\Omega$ indicates the unique $SO(3)$-invariant Haar measure on $S^2$ with $\int_{S^2} d\Omega =4\pi$, $\gB(\text{Sym}^N(\mathbb{C}^2))$ denotes the algebra of bounded operators on the symmetrized tensor product $\text{Sym}^N(\mathbb{C}^2)\subset \bigotimes_{n=1}^N\bC^2$, and $\Psi_N^\Omega$ are the {\bf spin coherent states} constructed in the following manner.
The vector space of the $N$-fold symmetric tensor product on $\bC^2$, i.e. $\text{Sym}^N(\mathbb{C}^2)\subset \bigotimes_{n=1}^N\bC^2$ has dimension equal to $N+1$. Using the bra-ket notation, let  $|\!\uparrow\rangle, |\!\downarrow\rangle$ be the eigenvectors of $\sigma_3$ in $\mathbb{C}^2$, so that
$\sigma_3|\!\uparrow\rangle=|\!\uparrow\rangle$ and $\sigma_3|\!\downarrow\rangle=- |\!\downarrow\rangle$, and where  $\Omega \in {S}^2$, with  polar angles  
$\theta_\Omega \in (0,\pi)$, $\phi_\Omega \in (-\pi, \pi)$, we then define the unit vector
\begin{align}\label{om1}
|\Omega\rangle_1= \cos \frac{\theta_\Omega}{2} |\!\uparrow\rangle + e^{i\phi_\Omega}\sin   \frac{\theta_\Omega}{2} |\!\downarrow\rangle.
\end{align}
 If $N \in \mathbb{N}$, the associated {\bf $N$-coherent spin state} $\Psi_N^{\Omega}:=|\Omega\rangle_N\in \mathrm{Sym}^N(\mathbb{C}^2)$, equipped with the usual scalar product $\langle \cdot ,\cdot \rangle_N$ inherited from    
$(\mathbb{C}^2)^N$,  is defined as follows \cite{Pe72}:
\begin{align}\label{om2}
|\Omega\rangle_N=  \underbrace{|\Omega\rangle_1 \otimes \cdots \otimes |\Omega \rangle_1}_{N \: times}.
\end{align}

The result stated in the following proposition provides a coherent pure state quantization of  $S^2$.

\begin{proposition}
Let $I=(1/\mathbb{N}_*)\cup\{0\}$ (with the  topology inherited from $[0,1]$),  where $\mathbb{N}_*=1,2,3,\cdots$, $\hbar=1/N$ ($N\in\mathbb{N}_*$). Define $\mathcal{H}_\hbar=\text{Sym}^N(\mathbb{C}^2)$. Denote by $\mu_L$ the Liouville measure on $S^2$ which coincides with the spherical measure $\sin{\theta}d\theta d\phi$ ($\theta\in (0,\pi) \ \phi\in (0, 2\pi)$). For any $\Omega:=(\theta,\phi)\in S^2$ define a unit vector $\Psi_N^\Omega\in \mathcal{H}_\hbar$ by \eqref{om2}. Denote the projection of $\Psi_N^\Omega\in\mathbb{S}\mathcal{H}_\hbar$ to $\mathbb{P}\mathcal{H}_\hbar$ by $\psi_N^{\Omega}$. 
Then the choices, 
\begin{align}
&q_\hbar(\Omega)=\psi_N^\Omega;\\
& c(\hbar)=(N+1)/4\pi;
\end{align}
so that $\mu_{1/N}(\theta,\phi)=\frac{N+1}{4\pi}\sin{\theta}d\theta d\phi$ yield a coherent pure state quantization of $S^2$ on $I$. For $f\in C(S^2)$ the associated Berezin quantization maps $Q_{1/N}^B(f)$ is defined by \ref{defquan3}.
\hfill $\blacksquare$
\end{proposition}
Moreover, as a result of \cite{MV} the maps \eqref{defquan3} define a surjection of $C(S^2)$ onto $\gB(\text{Sym}^N(\mathbb{C}^2))\cong M_{N+1}(\bC)$, the complex vector space of $(N+1)\times (N+1)$ matrices.

Finally, it can be shown that this family of Berezin quantization maps generates a continuous bundle of $C^*-$algebras whose continuous cross-sections are given by all sequences $(H_{1/N})$ for which $H_{1/N}\in\gB(\text{Sym}^N(\mathbb{C}^2))$ and $h_0\in C(S^2)$, and the sequence satisfies the following norm asymptotic equivalence property \cite[Thm. 8.1]{Lan17}
\begin{align}\label{normapproxasym}
\lim_{N\to\infty}||H_{1/N} - Q_{1/N}^B(h_0)||=0. 
\end{align}

\section{Quasi-locality, quasi-symmetric sequences and strict deformation quantization}\label{Qss}
We consider the standard $C^*$-inductive system over the projective tensor product over $I:=1/\mathbb{N}_*\cup\{0\}$. For this purpose,  we take any unital $C^*$-algebra $\gA$ and set
\begin{align}\label{Eq: quasi-local bundle}
    \mathcal{A}_{1/N}:=\begin{dcases}
        \gA^{N}\quad N\in\mathbb{N};
        \\
        [\gA]^\infty\quad N=\infty,
    \end{dcases}
\end{align}
where $\gA^N:=\gA^{\otimes N}$ denoted the projective tensor product of $\gA$ with itself. The $C^*-$algebra $[\gA]^\infty$ in turn is constructed in terms of the standard embedding maps $\varphi_N^M:\gA^M \hookrightarrow \gA^N$ with $N\geq M$ defining {\em local sequences}
\begin{align}
   &\varphi^M_N\colon \gA^M\ni A_M
    \mapsto A_M\otimes I^{N-M}\in \gA^N, \label{localsequence}
\end{align}
whose notion extends to {\em quasi-local sequences} (see \cite{OA2,BR1,BR2}). More precisely, the space  $[\gA]^\infty$ is defined by the quotient
\begin{align}\label{quotient3}
[\gA]^\infty:=\{(A_{1/N})_N\ | \ (A_{1/N})_N \ \text{quasi-local sequence}\}/\sim,
\end{align}
where $(A_{1/N})\sim (A_{1/N}')$ if and only if $\lim_{N\to\infty}||A_{1/N} - A_{1/N}'||=0$. Elements of $[\gA]^\infty$ are therefore equivalence classes, and the ensuing norm is given by
\begin{align*}
||[A_{1/N}]_N||:=\lim_{N\to\infty}||A_{1/N}||,
\end{align*}
so that $[\gA]^\infty$ is the completion of the space of these equivalence classes in this norm, dubbed {\em quasi-local algebra}.
Additionally, the fibers $\mathcal{A}_{1/N}$ and $\mathcal{A}_0$ can be shown to constitute a continuous bundle of $C^*$- algebras $\mathcal{A}^{q}$ over the base space $I=\{0\}\cup 1/\mathbb{N}_*\subset[0,1]$ (with relative topology, so that $(1/N)\to 0$ as $N\to\infty$).\footnote{The superscript $q$ occurring in $\mathcal{A}^{q}$ indicates that this continuous bundle of $C^*$-algebras corresponds to a non-commutative $C^*$-algebra of observables used to describe the thermodynamic limit.} The corresponding continuous cross-sections are the quasi-local sequences.
\\\\
We now extend the previous construction to the case of additional permutation symmetry. We set
\begin{align}\label{Eq: quasi-local bundle}
    \mathcal{A}_{1/N}^\pi:=\begin{dcases}
        S_N(\gA^{N})\quad  N\in\mathbb{N};
        \\
        [\gA]_\pi^\infty\quad N=\infty,
    \end{dcases}
\end{align}
where $S_N$ is the \textbf{symmetrization operator} defined by continuous and  linear extension on elementary tensors
\begin{align}\label{def:SN}
    S_N(A_1\otimes\ldots\otimes A_N):=
    \frac{1}{N!}\sum_{\pi}A_{\pi(1)}\otimes\cdots A_{\pi(N)}\,,
\end{align}
where the summation is over the elements $\pi$ in the permutation group of order $N!$. The space $[\gA]_\pi^\infty$ is defined in a similar was as before, that is, by the quotient \eqref{quotient3}, where in this case  ``quasi-local sequences'' is replaced by {\em quasi-symmetric sequences}.
To construct these we need to generalize the definition of $S_N$. For $N\geq M$ define a bounded  operator  $S_{M,N}: {\gA}^M 	\to {\gA}^N$, defined by linear and continuous extension of  
\begin{align}S_{M,N}(A) = S_N(A \otimes \underbrace{I \otimes \cdots \otimes I }_{N-M \mbox{\scriptsize times}}),\quad A \in {\gA}^M. \label{defSMN}
\end{align}
Clearly, $S_{N,N}=S_N$. Now, a sequence
 $(A_{1/N})_{N\in\mathbb{N}}$ is called {\bf symmetric}
if there exist $M \in \mathbb{N}$ and $A_{1/M} \in {\gA}^{\otimes M}$ such that 
\begin{align}\label{one}
    A_{1/N} = S_{M,N}(A_{1/M})\:\mbox{for all }  N\geq M,
\end{align}
and {\bf quasi-symmetric} if  $A_{1/N} = S_{N}(A_{1/N})$ if $N\in \mathbb{N}$,
and for every $\epsilon > 0$, there is a symmetric sequence $(A_{1/N}')_{N\in\mathbb{N}}$
as well as  $M \in\mathbb{N}$ (both depending on $\epsilon$) such that 
\begin{align}\label{quasisymmnew}
 \|A_{1/N}-A_{1/N}'\| < \epsilon\: \mbox{ for all } N > M.
\end{align} 
Furthermore, for any quasi-symmetric sequence the following limit exists
\begin{align}
a_0(\omega)=\lim_{N\to\infty}\omega^N(A_{1/N})\label{8.46},
\end{align}
where $\omega\in S({\gA})$, and $\omega^N=  \underbrace{\omega\otimes \cdots \otimes \omega}_{N\: \mbox{\scriptsize times}} \in S({\gA}^{\otimes N})$
is  the unique (norm) continuous linear extension of the following map that is defined on elementary tensors
\begin{align}
\omega^N(A_1\otimes\cdot\cdot\cdot\otimes A_N)=\omega(A_1)\cdot\cdot\cdot\omega(A_N).\label{linearstateextension}
\end{align}
The limit in \eqref{8.46} defines a function in $C(S(\gA))$ provided that $(A_{1/N})_{N\in\mathbb{N}}$ is quasi-symmetric, otherwise it may not exist. 

In fact, a non-trivial result shows that $[\gA]_\pi^\infty$ is commutative and isomorphic to $C(S(\gA))$, i.e. the $C^*-$algebra of continuous functions on the state space of $\gA$. Also in this case we have a continuous bundle of $C^*-$algebras $\mathcal{A}^{c}$ over the same base space $I=\{0\}\cup 1/\mathbb{N}_*\subset[0,1]$, but as opposed to the previous bundle, the $C^*-$ algebra at $N=\infty$ satisfies $[\gA]_\pi^\infty\cong C(S(\gA))$.\footnote{The superscript $c$ occurring in $\mathcal{A}^{c}$ indicates that this continuous bundle of $C^*$-algebras corresponds to a commutative $C^*$-algebra of observables of infinite quantum systems used to describe classical thermodynamics as a limit of quantum statistical mechanics.} The continuous cross sections of the bundle are the quasi -symmetric or macroscopic sequences.
\\\\
It turns out that the $C^*$-bundle $\mathcal{A}^{c}$ relates to a strict deformation quantization  the state space $S(\gA)$ in the case that $\mathfrak{A}=M_k(\mathbb{C})$ \cite{LMV}. Indeed, $X_k:=S(M_k(\mathbb{C}))$ is shown to admit a canonical Poisson structure. The quantization maps $Q_{1/N}$ to be defined on a dense Poisson algebra $\Tilde{\mathcal{A}}_0^\pi\subset\mathcal{A}_0^\pi:=C(S(M_k(\mathbb{C})))$ are defined as follows. First, our Poisson subalgebra $\Tilde{\mathcal{A}}_0^\pi$ is made of the restrictions to $S(M_k(\mathbb{C}))$ of polynomials in $k^2-1$ coordinates of $\mathbb{R}^{k^2-1}$. As each elementary symmetrized tensor of the form $T_{j_1}\otimes_s\cdot\cdot\cdot\otimes_sT_{j_L}$ (where $iT_1,\ldots, iT_{k^2-1}$ form a basis of the Lie algebra of $SU(k)$) may be uniquely identified with a monomial $p_L$ of degree $L$, one is allowed to define the quantization map $Q_{1/N}$. More precisely,  if $$p_L(x_1,\ldots, x_{k^2-1}) = x_{j_1} \cdots x_{j_L}\quad \mbox{where $j_1,\ldots, j_L \in   \{1,2,\ldots, k^2-1\}$,}$$ 
the quantization maps $Q_{1/N}: \Tilde{\mathcal{A}}_0^\pi\subset \mathcal{A}_0^\pi \to S_N(M_k(\bC)^N)$ act as 
\begin{align}\label{deformationquantization1}
 Q_{1/N}(p_L) &=
\begin{cases}
    S_{L,N}(T_{j_1}\otimes_s\cdot\cdot\cdot\otimes_sT_{j_L}), &\ \text{if} \ N\geq L \\
    0, & \ \text{if} \ N < L,
\end{cases}\\
Q_{1/N}(1) &= \underbrace{I_k \otimes \cdots \otimes I_k}_{\scriptsize N \: times}. \label{deformationquantization2},
\end{align}
and more generally they are defined as  the unique continuous and linear extensions of the written maps. It can be shown that the  maps $Q_{1/N}$ satisfy all the axioms of Definition \ref{def:deformationq}, yielding a strict deformation quantization of $X_k$. It follows directly from this construction that $X_k$ does not admit a (coherent) pure state quantization.

\section*{Acknowledgments} This author is supported by a postdoctoral fellowship granted by the Alexander von Humboldt-Stiftung (Germany). The author is grateful to Nicolò Drago for his feedback and illuminating discussions. The author also thanks the referee for the feedback and indicating points for improvement.

\end{document}